\documentclass[12pt]{article}
\usepackage[utf8]{inputenc}
\usepackage[T1]{fontenc}
\usepackage{amsmath,amssymb,amsthm,mathtools}
\usepackage{enumerate}
\usepackage[margin=1.25in]{geometry}
\usepackage[colorlinks=true, linkcolor=black, citecolor=black, urlcolor=black]{hyperref}
\usepackage{natbib}
\usepackage{booktabs}
\usepackage{tocloft}

\theoremstyle{plain}
\newtheorem{theorem}{Theorem}
\newtheorem{lemma}{Lemma}
\newtheorem{claim}{Claim}
\newtheorem{proposition}{Proposition}
\theoremstyle{definition}
\newtheorem{definition}{Definition}
\newtheorem{example}{Example}
\theoremstyle{remark}
\newtheorem{remark}{Remark}

\theoremstyle{plain}
\newtheorem{atheorem}{Theorem}
\newtheorem{alemma}[atheorem]{Lemma}

\theoremstyle{definition}
\newtheorem{adefinition}{Definition}

\theoremstyle{remark}
\newtheorem{aremark}{Remark}

\newcommand{\Thst}{\Theta^{\ast}}
\newcommand{\SU}{\mathit{SU}}
\newcommand{\CB}{C^{B}}   
\newcommand{\CN}{C^{N}}  
\newcommand{\PO}{\mathit{PO}}

\begin{document}

\title{The Paradox of Strategic Altruism}

\author{Foivos Savva\thanks{School of Economic, Social and Political Sciences,
University of Southampton, Southampton, UK.
E-mail: \texttt{F.Savva@soton.ac.uk}.}
\and
Michele Lombardi\thanks{University of Liverpool Management School,
University of Liverpool, Liverpool, UK; and Department of Economics
and Statistics, University of Napoli Federico~II, Naples, Italy.
E-mail: \texttt{michele.lombardi@liverpool.ac.uk}.}
\and Ritesh Jain\thanks{University of Liverpool Management School,
University of Liverpool, Liverpool, UK.
E-mail: \texttt{ritesh.jain@liverpool.ac.uk}.}}

\date{\today}

\maketitle

\thispagestyle{empty}
\clearpage

\begin{abstract}
\noindent
Can altruism serve as the behavioral foundation for efficient
institutional design? We study full implementation in Berge
equilibrium, the solution concept that formalizes
\emph{strategic altruism}: each agent's opponents collectively
maximize her payoff. While Berge
equilibrium resolves the Prisoner's Dilemma and can eliminate
cooperation failures in social dilemmas, we show that this optimism does
not survive the move from individual behavior to institutional
design. First, we show that the weak Pareto optimal rule is not implementable in Berge equilibrium. Second, on the unrestricted domain of strict preferences, any
social choice rule that is both weakly Pareto efficient and
implementable in Berge equilibrium must be dictatorial.
Finally, we show that Berge implementability is \emph{strictly more
demanding} than Nash implementability: every social goal
achievable under strategic altruism is also achievable under
self-interest, but not vice versa. The very feature that makes Berge equilibrium appealing
at the behavioral level, that is, opponents protecting each agent's
payoff, is precisely what forecloses efficient, non-dictatorial
institutions at the design level.
\end{abstract}
 
\clearpage

\thispagestyle{empty}
\tableofcontents

\newpage

\section{Introduction}
\label{sec:intro}
 
Individual rationality and collective welfare are frequently in
conflict. A farmer who draws from a shared aquifer maximizes her
own yield, and so does each of her neighbors, with the result
that the aquifer is depleted and everyone is left worse off than
if they had all drawn less. A country that cuts carbon emissions
bears the cost alone while sharing the benefit with the whole
world, and so every country waits for the others to act first. In each case,
behavior that is individually sensible produces outcomes that are
collectively disastrous. Game theory distills this tension into
the Prisoner's Dilemma: two players, each choosing whether to
cooperate or defect, each find that defection is the dominant
strategy. Yet, mutual defection leaves both strictly worse off
than mutual cooperation. The Prisoner's Dilemma is not merely
an abstract example; it underlies a vast range of real collective
problems, from arms races and climate negotiations to cartel
instability, public goods provision and union membership decisions \citep{olson1971logic}.
 
One response to this predicament is institutional: design rules,
procedures, and mechanisms that give individuals the incentive to
behave cooperatively. A second response is behavioral: change the
preferences or norms that govern individual action. The first
response is the domain of mechanism design. The second is the
domain of social psychology, behavioral economics, and moral
philosophy, and it is the starting point of this paper.

The most natural behavioral escape from the social dilemma is
altruism. If each agent genuinely cared about the welfare of
others, the dilemma would dissolve.
An agent who maximizes the welfare of her opponents has no
incentive to defect since defection harms the others, which is
precisely what an altruistic agent wishes to avoid. The unique
equilibrium of the Prisoner's Dilemma under this behavioral norm is mutual cooperation,
which is also Pareto-optimal. Altruism solves the
social dilemma not by changing the rules of the game, but by
changing how players play it.
 
This logic has been given a precise game-theoretic formulation
by \citet{berge1957theorie}, and developed formally by
\citet{zhukovskiy1985} and \citet{COLMAN2011166}. A
\emph{Berge equilibrium} is a strategy profile in which each
player's \emph{opponents} collectively maximize \emph{that
player's} payoff, holding her strategy fixed. Where Nash
equilibrium asks what each player does for herself, Berge
equilibrium asks what each player's opponents do for her. We call this type of behavior \textit{strategic altruism}, where the altruistic motive emerges at the strategic level, as an equilibrium notion, and not at the preference level.\footnote{There are many well-known models of altruism or other-regarding preferences, such as \cite{fehr1999theory} or \cite{bolton2000erc}, among others.} The behavioral norm it encodes---do the best you can for
others, given what they do---can be interpreted as a strategic
formalization of the Golden Rule that appears across cultures
and moral traditions \citep{haller2024berge}. In the Prisoner's
Dilemma, it delivers mutual cooperation. In coordination games
such as the Stag Hunt, it resolves the payoff-dominance problem.
A growing literature has shown that this cooperative equilibrium
concept is both theoretically well-founded and empirically
plausible as a description of altruistic behavior
\citep{COLMAN2011166, sugden2011mutual, UNVEREN20231,
courtois2015play, crettez2017berge_sugden}.

Behavioral change may nevertheless not be enough. Societies do not rely solely on individuals being altruistic. They design
\emph{institutions} (constitutions, voting rules, regulatory
mechanisms, allocation procedures and so on) to achieve desirable social goals, by taking the behavioral norms of individuals into account. Now, while Berge altruism can support attractive social outcomes in certain types of strategic interactions such as social dilemmas, this does not necessarily imply that it is good for the collective in any strategic situation. In other words, while Berge equilibrium as a standard of behavior seems to perform well \textit{locally}, its implications for \textit{global} institution design are unclear. This is especially crucial since, in general, institutions are expected to implement the desired social goal, for any configuration of individual preferences, that is, states of the world. This type of exercise is the
domain of \emph{implementation theory}, where a hypothetical social planner wishes to implement a social goal, without knowing the exact individual preferences. Its central
conclusion is that what a society can achieve institutionally
depends critically on how agents behave within the game
\citep{maskin1999nash, palfrey2002implementation}.
 
This observation motivates a question that, to our knowledge,
has not been asked before. If agents are altruistic in the Berge
sense, that is, if each player's opponents maximize her
payoff, can a mechanism designer exploit this behavioral
disposition to implement efficient and fair social goals? The present paper shows,
definitively and perhaps surprisingly, that the answer is no.

We study full implementation in Berge equilibrium, the
requirement that, in every possible configuration of individual preferences, the
set of Berge equilibrium outcomes of a mechanism coincides
\emph{exactly} with the set of socially optimal outcomes
prescribed by the planner's goal. Our main result is a
dictatorship theorem:
 
\medskip
\noindent\textbf{Main Result} (informal). \textit{On the
unrestricted domain of strict preferences, any social choice
rule that is both weakly Pareto efficient and implementable in
Berge equilibrium must be dictatorial: it always selects the
most-preferred outcome of a single fixed agent, regardless of
everyone else's preferences.}
\medskip
 
Weak Pareto efficiency is the mildest efficiency requirement
one can impose: it merely rules out outcomes that every agent
unanimously prefers to abandon. The unrestricted domain
imposes no restrictions on preferences: agents may rank the
alternatives in any order. The theorem therefore says that
altruistic agents, whatever their preferences, cannot be
governed in an efficient and non-dictatorial manner. The
only efficient rules that a society of strategic altruists can implement
are those that simply impose one agent's will on everyone else. This is the paradox: strategic altruism at the individual level, far from
enabling better collective institutions, forecloses them.
Agents who unconditionally serve others' interests cannot be
organized into an efficient and democratic institution. 
 
We establish two further results that complete the picture.
First, not only the full Pareto correspondence but also natural
efficient non-dictatorial rules, including a simple ``king-maker''
rule in which one agent designates another as decision-maker, are
non-implementable in Berge equilibrium, even though they are
implementable in Nash equilibrium \citep{maskin1999nash}. Second,
we show that Berge implementability is \emph{strictly more
demanding} than Nash implementability: every social choice rule
implementable under Berge equilibrium is also implementable under Nash
equilibrium, but not vice versa. This is the \emph{altruism tax}:
switching from self-interested to altruistic strategic reasoning
strictly narrows the menu of social goals available to the
institution designer. A society of sinners gives the planner more
leverage than a society of saints.

As a by-product of this last result, we provide a complete characterization of the social
choice rules that are implementable in Berge equilibrium, essentially the strategic-altruism
counterpart of the \citet{moore1990nash} characterization of Nash implementation. The
characterization isolates the precise structural property that separates the two notions: a
cross-profile requirement which, under self-interest, binds only when there are two agents,
but which under strategic altruism binds for every group size. This is essentially the formal source of
the altruism tax.


The result is surprising on two counts. First, one might have
expected altruistic agents to be easier to coordinate: they are
not trying to game the mechanism for personal gain, they are
genuinely trying to help each other. The impossibility shows
that the very feature that makes Berge equilibrium appealing at
the behavioral level, that is, each agent's opponents working to
maximize her payoff, is precisely what makes it destructive at
the institutional level. In Berge equilibrium, each agent is
``protected'' by her opponents: they cannot, in equilibrium,
push the outcome below her current payoff. When the planner
tries to design a mechanism that selects good outcomes in many
different states of the world, these protections generate
contradictory obligations across scenarios. The mechanism cannot
simultaneously honor every agent's protection in every state,
and efficiency then forces all these obligations to collapse into
the preferences of a single agent, the dictator.
 
Second, the result is stronger than its closest classical
antecedent. \citet{hurwicz1978construction} showed that
on the full domain of strict preferences, any Pareto-efficient
and Nash-implementable rule must be dictatorial when there are
only two agents. For three or more agents, this Nash
impossibility \emph{disappears}: efficient, non-dictatorial
rules, like Maskin's king-maker rule, become
Nash-implementable \citep{maskin1999nash}. Our theorem shows
that under Berge equilibrium the impossibility persists for any
number of agents, however large. The behavioral assumption is
what drives the result, not the group size.
 
The conceptual moral is equally sharp. The social dilemma
literature and the mechanism design literature address the
same problem from different levels. Berge equilibrium resolves
the social dilemma at the \emph{behavioral} level, by changing
how agents play a given game. However, while it solves the social dilemma locally, it creates insurmountable issues at the institutional level, globally, as it does not help us design institutions that are both efficient and non-dictatorial. The
impossibility is not a failure of altruism as an ethical value;
it is a structural property of the mechanism design exercise. On this, one is reminded of
Adam Smith's observation that individual self-interest can serve
the social good more reliably than deliberate benevolence:
 
\begin{quote}
By pursuing his own interest he frequently promotes that of the
society more effectually than when he really intends to promote it.
\citep{smith1776wealth}
\end{quote}
 
Our results offer an institutional analog of this exact insight. Naturally, though, Smith's seminal work focused on the relatively new institution of free markets, which had been expanding in Great Britain during the 18\textsuperscript{th} and 19\textsuperscript{th} century. Thus, his observation that selfish behavior can more reliably promote the public good points directly to the existing institutional setting. By utilizing the implementation theory framework, our results carry an even stronger message: \textit{even under optimal institutional arrangements for each behavioral norm, individualistic egoism is never dominated by strategic altruism, from the collective good point of view}.

The paper is organized as follows. Section \ref{sec:literature} places the paper in the literature. Section~\ref{sec:model} sets
out the implementation theory model,
Berge equilibrium concept, and the cross-profile construction that
underpins all our results. Section~\ref{sec:impossibility}
presents the impossibility results: the non-implementability of
the weak Pareto rule and the king-maker rule, as well as the dictatorship
theorem. Section~\ref{sec:berge-implies-nash} establishes that
Berge implementability is strictly more demanding than Nash
implementability and develops the altruism tax interpretation.
Section~\ref{sec:conclusion} concludes and discusses extensions.

\section{Related Literature}
\label{sec:literature}

\paragraph{Implementation theory.}
The implementation problem was formalized by \citet{hurwicz1972},
who asked whether a planner can design a game whose equilibrium
outcomes coincide with the ones recommended by a given social choice rule even when individual
preferences are unknown. \citet{maskin1999nash}, in his influential paper, provided the
foundational answer in a remarkably general social choice environment. He showed
that monotonicity is \emph{necessary} for Nash implementability: if an outcome is
selected in one state, it must continue to be selected in any state where it has
not moved down in any agent's ranking. He then showed that, for three or more
agents, monotonicity together with the auxiliary condition of no-veto power (if
all but one agent top-rank an outcome, then it must be selected) is \emph{sufficient}.
 
A complete characterization was provided by \citet{moore1990nash}, who closed the
gap between these necessary and sufficient conditions: Nash implementability is
equivalent to their Condition $\mu$ for three or more agents, and to their
Condition $\mu_2$ for two agents (see also \citealt{dutta1991necessary} for the
two-agent case). The two-agent case is significantly harder, and it is there that
the closest antecedent to our dictatorship theorem appears:
\citet{hurwicz1978construction} showed that with two agents on the
unrestricted domain of strict preferences, any weakly Pareto efficient and Nash-implementable
rule must be dictatorial. For three or more agents, however,
this impossibility disappears, and efficient, non-dictatorial
rules become Nash-implementable. Our main theorem shows that, in the Berge case, this impossibility persists for any number of agents.
 
While the full implementation literature has explored many extensions of
the Nash benchmark,\footnote{For example, strong Nash implementation \citep{dutta1991implementation,sang1996algorithm}, subgame-perfect implementation
\citep{moore1988subgame}, virtual implementation
\citep{abreu1991virtual}, Bayesian implementation
\citep{jackson1991bayesian}, and implementation in
undominated strategies \citep{palfrey1990undominated} among others. See also \cite{kartik2012implementation}, or \cite{gavan2025safe} for more recent contributions.} most, if not all of the contributions, operate under the assumption that agents maximize their own gain. On the contrary,
we pursue a relatively unexplored path: we change the behavioral assumption placed on
agents from selfishness to strategic altruism.\footnote{There is, however, a related literature on full implementation without rationality assumptions. See \cite{de2014behavioral}, \cite{hayashi2023behavioral}, or \cite{barlo2023behavioral}.} This approach is closer in spirit to \citet{roemer2010kantian},
who studied implementation under Kantian equilibrium, but
arrives at sharply different conclusions: in economic environments, the proportional solution, which is efficient and non-dictatorial, is implementable in Kantian equilibrium.

The consequences of other-regarding preferences have nevertheless been explored in partial implementation settings, with mostly positive results. \cite{kucuksenel2012behavioral} shows that increasing the level of altruism among agents tends to increase the level of public good provision towards its efficient level. \cite{kozlovskaya2019public} study strategy-proofness with reciprocal agents and show that with a sequential pivot mechanism, efficient provision can be implemented in a strategy-proof manner. \cite{bierbrauer2016mechanism} assume intention-based social preferences based on the model of \cite{rabin1993incorporating} and show that, when the weight of kindness in agents' utility functions is commonly known, Pareto efficient outcomes are implementable.\footnote{See also \cite{bierbrauer2017robust}.} On the contrary, in our setting, Berge implementation is in general incompatible with efficiency.

Finally, our results are also related to the literature on implementation with pro-social motives as ``tie-breaking rules'', that is, motives that trigger when agents are indifferent between outcomes. The most well-studied one is partial honesty \citep{matsushima2008role,dutta2012nash, korpela2014bayesian,  lombardi2020partially,savva2018strong}. A partially-honest agent is an agent who deceives the mechanism designer only when
the truth poses some obstacle to her material well-being. Thus, she does not deceive when
the truth is equally efficacious.\footnote{For other types of tie-breaking rules, see \cite{kimya2017nash}, or \cite{savva2021motives}.}  In general, the results from this exercise are quite permissive from the design point of view, in that they can expand the set of implementable rules compared to Nash implementation. By contrast, we show that when agents are strategic altruists, the set of such rules shrinks.

\paragraph{Berge and other ``cooperative'' equilibrium concepts.}
Berge equilibrium was introduced by
\citet{berge1957theorie} and formalized for differential games in \citet{zhukovskiy1985}.\footnote{For some more recent existence results, see \cite{haller2025existence}, or \cite{riedel2026berge}.} More recently, Berge equilibrium has attracted attention through
\citet{COLMAN2011166}, who explore its implications in social dilemmas and related applications, arguing for its compellingness in capturing the essence of altruism. More results on the normative dimension of Berge equilibrium can be found in \citet{salukvadze2020berge} and
\citet{haller2024berge}.

The relationship between Berge and other equilibrium notions 
has been studied in several directions. \citet{courtois2015play} endogenize whether players act like Nash or Berge and study the consequences of such situational approach for some social dilemmas. \citet{crettez2017berge_sugden} links Berge equilibrium to
Sugden's notion of mutually beneficial practice \citep{sugden2011mutual,
sugden2015team}, showing that Berge equilibria can be
interpreted as the strategic foundation of team reasoning in
some classes of games. Finally, \cite{crettez2019unilateral}, \cite{crettez2020existence} and \cite{schouten2019unilateral} study unilateral support equilibria, where each agent ``supports'' some player and provide connections with Berge equilibrium. 

Kantian equilibrium was introduced in \cite{roemer2010kantian} and \cite{roemer2015kantian} and expanded in \cite{roemer2019we}. The idea behind Kantian equilibrium is that when a player considers a possible deviation from a given strategy profile, she assumes that everyone will do so as well in a proportional manner, thus incorporating a notion of Kant's categorical imperative into agents' strategic concerns. In \citet{roemer2021socialism},
Kantian equilibrium is presented as an alternative behavioral foundation
for socialist institutions, interpreting cooperative behavior
as an ethos compatible with decentralized production. Contrary to Berge though, in a Kantian equilibrium, each agent maximizes her own payoff. On this, \citet{UNVEREN20231} show
that in symmetric games efficient Berge equilibria are generically also Kantian
equilibria, establishing a deep
connection between other-regarding and rule-based cooperation. This also implies that even though the two
norms diverge
sharply in their institutional properties, they are observationally equivalent in social dilemmas.

Our results also relate to more recent studies on the theory of cooperation and its evolutionary foundations, as in \cite{alger2016evolution}, \cite{alger2019evolutionary}, or \cite{alger2020evolution}. Even more related to Berge equilibrium,
\citet{alger2016evolution} show that a class of preferences
combining self-interest with a probabilistic version of Kant's
categorical imperative which they call \emph{Homo moralis}, is
evolutionarily stable under assortative matching. Taken together with the result of \citet{UNVEREN20231}, this implies that in symmetric social dilemmas, sufficiently assortative matching selects through
evolutionary pressure the same cooperative behavior that Berge
equilibrium prescribes. Evolution thus provides a microfoundation
for the Berge behavioral assumption that does not rely on agents
consciously following the Golden Rule.

\paragraph{Moral and other-regarding preferences.} Finally, our results are related to the vast literature on other-regarding preferences, as in \cite{fehr1999theory}, \cite{bolton2000erc}, or \cite{andreoni1990impure}. However, these papers incorporate motives such as inequality-aversion, reciprocity, or altruism at the \textit{preference level}, while our contribution studies other-regarding motives and altruism specifically, at the \textit{strategic level}. We feel this approach is relatively unexplored, and we see our (mostly negative) results as a first step toward incorporating such other-regarding strategic concerns in mechanism design.

\section{The Model}
\label{sec:model}
 
\subsection*{Society, Outcomes, and Preferences}
\label{subsec:primitives}
 
There is a finite set of \emph{agents} $N = \{1, \ldots, n\}$, with
$n \geq 2$. One can think of them as voters in a committee, partners in a
firm, or citizens subject to a collective decision. There is also a
finite set of \emph{outcomes} $X$, which can be interpreted as the set of electoral candidates, allocations,
or social alternatives in general.
 
Agents care about which outcome is selected, and their preferences
may depend on circumstances such as the state of technology. This is modeled by saying that agent $i$'s preferences depend on the \emph{state of the world} $\theta$ drawn from a set $\Theta$. In
state $\theta$, agent $i$ holds a preference ordering\footnote{That is, a complete and transitive binary relation on $X$.} $R_{i}(\theta)$
over $X$: we write $x \, R_{i}(\theta) \, y$ to mean that $i$ weakly
prefers $x$ to $y$ in state $\theta$, $x \, P_{i}(\theta) \, y$ for
the strict preference ($xR_i(\theta)y$ and not $yR_i(\theta)x$) and $xI_i(\theta)y$ for indifference (both $xR_i(\theta)y$ and $yR_i(\theta)x$ are true).
A profile $R(\theta) = (R_1(\theta), \ldots, R_n(\theta))$ collects the
preferences of all agents. For non-empty sets $A, B \subseteq X$ we write
$A \, P_{i}(\theta) \, B$ to mean that $a \, P_{i}(\theta) \, b$ for every $a \in A$ and
every $b \in B$.
 
A preference ordering is called \emph{strict} if no agent is ever
indifferent between two distinct outcomes. The \emph{unrestricted domain
of strict preferences}, denoted $\Thst$, consists of all states in which
every agent holds a strict ordering over $X$. Essentially, when $\Theta\supseteq\Thst$, we
do not restrict how agents strictly rank the alternatives: any such strict preference is under consideration.
 
Two sets derived from a preference ordering will appear
throughout the analysis. For any outcome $x$, state $\theta$, and agent
$i$, define
\begin{align*}
  L_{i}(x, \theta) &= \{ y \in X : x \, R_{i}(\theta) \, y \},\\
  \SU_{i}(x, \theta) &= \{ y \in X : y \, P_{i}(\theta) \, x \}.
\end{align*}
$L_{i}(x, \theta)$ is the \emph{lower-contour set} of $i$ at $x$: everything
$i$ finds weakly worse than $x$ in state $\theta$, including $x$ itself.
$\SU_{i}(x, \theta)$ is the \emph{strict upper-contour set}: everything
$i$ strictly prefers to $x$. The two sets partition $X$, with $x$ sitting
in $L_{i}(x, \theta)$ but not in $\SU_{i}(x, \theta)$. Reading the
lower-contour set of $i$ at $x$ as the outcomes that cannot improve upon $x$
for $i$ will be useful when we state the equilibrium conditions below.
 
\subsection*{The Planner's Problem}
\label{subsec:planner}
 
A hypothetical \emph{planner} (a constitution-writer, a market regulator, a committee
chair, or the set of agents at an ex-ante stage) has in mind a rule that selects her desirable outcomes in each state of the world. This mapping is the object she would
like to ``implement.''
 
\begin{definition}[Social choice rule]
A \emph{social choice rule} (SCR) is a correspondence
$F \colon \Theta \rightrightarrows X$ that assigns to each state $\theta$
a nonempty set $F(\theta) \subseteq X$ of \emph{socially optimal} outcomes.
\end{definition}
 
The SCR encodes the planner's normative objectives: it records, for every
possible configuration of preferences, which outcomes society should select.
Examples include the Pareto correspondence (select all outcomes that are not
unanimously dominated), majority-voting rules, and egalitarian allocations.
The planner would like to guarantee that the outcome chosen by society
actually belongs to $F(\theta)$, whatever $\theta$ turns out to be.

Moreover, a benevolent planner would potentially like $F$ to represent all agents' preferences to some extent, or would at least not want it to always coincide with the top outcomes of a single agent. A rule of the latter kind is \textit{dictatorial}: it is the most extreme case of power concentration, and we define it below:

\begin{definition}[Dictatorial SCR]
A SCR $F$ is \emph{dictatorial} if there exists an agent $i \in N$, called the
\emph{dictator}, such that for all $\theta \in \Theta$ and all $x \in X$,
\[
  x \in F(\theta) \iff X \subseteq L_{i}(x, \theta).
\]
Equivalently, in every state $\theta$, the rule selects agent $i$'s most-preferred
outcome.
\end{definition}

What makes the implementation problem non-trivial is that the planner does not observe $\theta$ directly.
Only the agents know the true state, and they may have incentives to
misrepresent it. The planner must therefore elicit the relevant information
through a procedure that agents are willing to participate in honestly, or
at least one whose equilibrium outcomes coincide with $F(\theta)$ regardless
of what the agents privately know.

\subsection*{Mechanisms}
\label{subsec:mechanisms}
 
The institutional procedure the planner designs is modeled as a
\emph{mechanism}: a game form in which agents submit messages and the
outcome is determined by those messages.
 
\begin{definition}[Mechanism]
A \emph{mechanism} $\Gamma = (M, g)$ consists of:
\begin{enumerate}[\upshape(i)]
\item a \emph{message space} $M = \prod_{i \in N} M_{i}$, where $M_{i}$
is the set of messages (or ``strategies'') available to agent $i$; and
\item an \emph{outcome function} $g \colon M \to X$ that maps every profile
of messages to an outcome.
\end{enumerate}
\end{definition}
 
The planner commits to the mechanism before the agents play it that is, she publishes
the rules of the game before any agent takes any action.
The agents then observe $\theta$ privately and simultaneously choose their
messages. The outcome $g(m)$ is then implemented.
 
A message profile in which only agent $i$ departs from some reference
profile $m$ by playing $m'_i$ is written $(m_{i}', m_{-i})$, where $m_{-i}$ collects the
messages of every agent other than $i$. Symmetrically, if all agents
\emph{except} $i$ deviate to $m'_{-i}$, the resulting profile is written
$(m_{i}, m_{-i}')$. Finally, let $g(M)=\{x\in X:x=g(m)$ for some $m\in M\}$ be the range of the mechanism.
 
\subsection*{Strategic Altruism: The Berge Equilibrium}
\label{subsec:berge}
 
The standard assumption in mechanism design is that agents are
purely \emph{self-interested}: each agent chooses the message that
maximizes her own payoff, taking everyone else's messages as fixed.
The resulting solution concept is Nash equilibrium, in which no
agent can improve her situation by unilaterally changing her message.
 
This paper studies the case when agents are \emph{altruistic}
in a specific and formally precise sense, as formalized by Berge equilibrium.
In a Berge equilibrium, it is not each agent who maximizes her own
payoff; it is each agent's \emph{opponents} who collectively
maximize \emph{her} payoff. Every player devotes her strategic
effort entirely to advancing others' interests.\footnote{An intuitive way of thinking about Berge equilibrium is that for each agent $i$, all others ``have got $i$'s back'' in that they collectively support $i$.} We call this
behavioral protocol \emph{strategic altruism} to distinguish it from other types of altruism that occur at the preference level.
 
\begin{definition}[Berge equilibrium]
\label{def:BE}
A message profile $m \in M$ is a (pure strategy) \emph{Berge equilibrium} at state
$\theta$ if, for every agent $i \in N$ and any message
profile $m_{-i}' \in M_{-i}$,
\[
  g(m) \, R_{i}(\theta) \, g(m_{i}, m_{-i}').
\]
\end{definition}
 
The condition says that, when we fix agent $i$'s own message at $m_i$, no
reallocation of her opponents' messages can produce an outcome that
$i$ strictly prefers to $g(m)$. In other words, $i$'s opponents
are already doing the best they can \emph{for $i$}. This is the
mirror image of Nash equilibrium.\footnote{For completeness, we define Nash equilibrium as well: A message profile $m\in M$ is a \textit{Nash equilibrium} at state $\theta$ if, for any $i\in N$ and $m'_i\in M_i$, we have $g(m)R_i(\theta)g(m'_i,m_{-i})$.} In Nash, it is $i$ herself who
cannot do better for herself; in Berge, it is $i$'s opponents who
cannot do better for $i$. We view Berge as standing on the other side of the ``behavioral spectrum'', opposite to Nash equilibrium, with Nash representing extreme selfishness and Berge representing extreme altruism. Let $BE(\Gamma,\theta)$ be the set of Berge equilibria and $NE(\Gamma,\theta)$ be the set of Nash equilibria of $\Gamma$ at state $\theta$.

\subsection*{Implementation}
\label{subsec:implementation}
 
We now describe the planner's design problem. She must construct a mechanism
$\Gamma$ such that, in \emph{every} state $\theta$, the set of equilibrium
outcomes of $\Gamma$ coincides exactly with $F(\theta)$. This consists of two requirements. \emph{Existence} says that every socially optimal outcome must be
supported as an equilibrium outcome: In equilibrium, the mechanism should not miss any
alternative that the planner considers optimal. \emph{Uniqueness} says that every equilibrium must
produce a socially optimal outcome: the mechanism should not select anything in equilibrium that
the planner does not consider optimal. Together they guarantee that the planner's goal is
implemented fully, in every state.\footnote{For arguments on why full implementation rather than partial, that is, only a subset of the socially optimal outcomes being supported as equilibrium ones, should be the desideratum, see \cite{thomson1996concepts}.}
 
\begin{definition}[Berge implementation]
A mechanism $\Gamma = (M, g)$ \emph{Berge-implements}\ the SCR $F$ if, for
every state $\theta \in \Theta$,
\begin{enumerate}[\upshape(i)]
\item[\upshape(i)] \emph{(Existence)} for every $x \in F(\theta)$, there
exists a Berge equilibrium $m \in M$ with $g(m) = x$; and
\item[\upshape(ii)] \emph{(Uniqueness)} for every Berge equilibrium $m$ at
$\theta$, $g(m) \in F(\theta)$.
\end{enumerate}
We say that $F$ is \emph{Berge-implementable} if such a $\Gamma$ exists.
\end{definition}
 
The contrast with Nash implementation is definitional: $\Gamma$
\emph{Nash-implements} $F$ if the same two conditions hold with Berge
equilibria replaced by Nash equilibria, and $F$ is \emph{Nash-implementable}
if such a $\Gamma$ exists. All the formal apparatus is identical and only
the behavioral assumption changes. This makes the comparison between the two
implementation notions clean, as any difference in the set of implementable SCRs
is attributable purely to the difference between strategic altruism and
self-interest.
 
\subsection*{Efficiency}
\label{subsec:efficiency}
 
The planner's social goal will typically incorporate some notion of efficiency.
 
\begin{definition}[Weak Pareto efficiency]
An outcome $x \in X$ is \emph{weakly Pareto dominated} at $\theta$ if there
exists $y \in X$ such that every agent strictly prefers $y$ to $x$:
$y \, P_{i}(\theta) \, x$ for all $i \in N$.
 
The \emph{weak Pareto rule} $\PO$ selects all outcomes that are not weakly
Pareto dominated:
\[
  \PO(\theta) = \bigl\{ x \in X : \text{there is no } y \in X \text{ with }
  y \, P_{i}(\theta) \, x \text{ for all } i \in N \bigr\}.
\]
A SCR $F$ satisfies the \emph{Pareto property} (is Pareto efficient, or simply efficient) if $F(\theta) \subseteq
\PO(\theta)$ for every $\theta$.
\end{definition}
 
An outcome outside $\PO(\theta)$ is one that \emph{every} agent strictly
prefers to abandon: everyone agrees it is not good enough. The Pareto
property demands no more than this, as it rules out only outcomes that
are unanimously condemned. It is, in this sense, the mildest efficiency
requirement one can impose on a social goal. The first result of this
paper shows that strategic altruism cannot meet even this minimal bar.

\subsection*{The Cross-Profile Construction}
\label{subsec:crossprofile}
 
Before we proceed to our results, we need a few pieces of extra notation. Specifically, the analysis that follows requires reasoning about what happens when
equilibrium strategies from \emph{different} states are combined. For this, given a SCR $F$, let
\[
  \Sigma^{F} = \bigl\{ (x^{i}, \theta^{i})_{i \in N} :
  x^{i} \in F(\theta^{i}) \text{ for all } i \in N \bigr\}
\]
be the set of \emph{cross-profile tuples}, that is, collections in which each
agent $i$ is paired with a state--outcome pair $(x^i, \theta^i)$
consistent with $F$. Each agent $i$'s pair comes from a possibly
different state: agent $1$ may be thinking about state $\theta^1$,
agent $2$ about state $\theta^2$, and so on. A way to think about this construction is as if each agent is making an ``announcement'': agent $1$ claims that the true state is $\theta^1$ and that $x^1$ should be implemented, agent $2$ makes a claim about $(x^2,\theta^2)$ and so on. 

Now, suppose that $\Gamma=(M,g)$ Berge-implements $F$ and let $m(x,\theta)$ be a Berge equilibrium in state $\theta$ such that $g(m(x,\theta))=x$. Then, for any $\sigma=(x^i,\theta^i)_{i\in N} \in \Sigma^F$, let $m^{\sigma}=\bigl(m_1(x^1,\theta^1),m_2(x^2,\theta^2),\ldots,m_n(x^n,\theta^n)\bigr)$ be the message profile in which each agent $i$ plays her part of a Berge equilibrium message profile $m_i(x^i,\theta^i)$, according to $\sigma$. Finally, given $\sigma\in \Sigma^F$, let:

\begin{center}
$O_i^B(x^i,\theta^i)=\bigl\{g\bigl(m_i(x^i,\theta^i),m'_{-i}\bigr):m'_{-i}\in M_{-i}\bigr\}$.
\end{center}

$O_i^B(x^i,\theta^i)$ has a useful interpretation. It is the set of outcomes that $i$'s opponents can induce given that $i$ has played her part of a Berge equilibrium according to $\sigma$. Write $w(\sigma)\equiv g(m^{\sigma})$ for the \emph{cross-profile outcome}. Note that $m^{\sigma}$ is a Berge equilibrium at a state $\theta$ if and only if $O^B_i(x^i,\theta^i)\subseteq L_i\bigl(w(\sigma),\theta\bigr)$ for every $i\in N$. The relevant benchmark is the outcome that $m^{\sigma}$ actually induces, not agent $i$'s own target $x^i$. These cross-profile constructions will be used in proving our main impossibility results.

\section{Impossibility Results}
\label{sec:impossibility}

We now establish the paper's central negative results. The organizing idea is that strategic
altruism imposes a structural constraint on every
implementing mechanism that is absent under Nash equilibrium. This constraint, isolated in
Lemma \ref{lem:cross} below, is the common source of both impossibility results.

The first result shows that even the mildest efficiency standard---weak Pareto
optimality---cannot be achieved by any mechanism when agents behave as strategic
altruists. This finding is already surprising: one might expect agents who actively look
after one another's welfare to have no difficulty selecting outcomes no one is unanimously
worse off with. The second and possibly more striking result shows that the difficulty runs much
deeper: on the unrestricted domain of strict preferences, the \emph{only} efficient social
goals that a strategically altruistic society can implement are \emph{dictatorial}. Both results
flow from a single, elementary property of Berge equilibrium that we isolate as a lemma before
proceeding to the main theorems.

\subsection{A Key Structural Constraint: The Cross-Profile Bound}
\label{subsec:crossprofbound}

The following Lemma isolates the only consequence of Berge implementation that the rest of
this section uses, and it concerns the cross-profile construction. It essentially says that when every agent plays the message from their own individual
equilibrium scenario, the resulting outcome must be weakly worse for each agent, than the
outcome in their own scenario.

\begin{lemma}[Cross-profile bound]\label{lem:cross}
Suppose $\Gamma = (M, g)$ Berge-implements $F$. Then, for every profile
$\sigma = (x^{i}, \theta^{i})_{i \in N} \in \Sigma^{F}$,
\[
  w(\sigma) \in \bigcap_{i \in N} L_{i}(x^{i}, \theta^{i}).
\]
In words: every agent weakly prefers her own target $x^{i}$ at state $\theta^i$ to the cross-profile outcome
$w(\sigma)$.
\end{lemma}

\begin{proof}
Fix any agent $i$. In the cross-profile message $m^{\sigma}\equiv (m^{\sigma}_i,m^{\sigma}_{-i})$, agent $i$ plays
$m_{i}(x^{i}, \theta^{i})$, while all $j\neq i$ (that is, her opponents) play $m_{j}(x^j,\theta^j)$. By the definition of
$O^B_{i}(x^i,\theta^i)$, the cross-profile outcome satisfies
$w(\sigma) \in O^B_{i}(x^i,\theta^i)$. Since $m(x^{i}, \theta^{i})$ is a Berge
equilibrium at $\theta^{i}$, $i$'s opponents cannot induce an outcome strictly better for $i$
than $x^{i}$, which means $O^B_{i}(x^i,\theta^i) \subseteq L_{i}(x^{i}, \theta^{i})$.
Thus, $w(\sigma) \in L_{i}(x^{i}, \theta^{i})$. Since $i$ was arbitrary, the intersection
bound follows.
\end{proof}

The Lemma says that the agents' individual targets always admit a common ``floor'': an
outcome $w(\sigma)$ that no agent strictly prefers to her own target. The Lemma imposes this
floor requirement as a \emph{necessary} condition on every Berge-implementable SCR. Under Nash equilibrium, no comparable constraint
exists: self-interested agents make no promise about what they will do
in \emph{other agents'} scenarios, so no cross-profile bound can be
derived. This is precisely why the two behavioral assumptions lead to
such different institutional possibilities.

\subsection{The Weak Pareto Rule Is Not Berge-Implementable}
\label{subsec:pareto-impossible}

In the following simple three-agent, three-outcome example, we show that the weak Pareto rule cannot be implemented under strategic altruism.

\begin{example}[Failure of the weak Pareto rule]\label{ex:wp}
Let $N = \{1, 2, 3\}$, $X = \{x, y, z\}$, and consider two states $\theta$ and $\theta'$ with
the following preference orderings (the most-preferred outcome listed first):

\begin{center}
\begin{tabular}{ccc@{\quad}ccc}
\multicolumn{3}{c}{State $\theta$} & \multicolumn{3}{c}{State $\theta'$}\\[3pt]
$R_{1}(\theta)$ & $R_{2}(\theta)$ & $R_{3}(\theta)$ &
$R_{1}(\theta')$ & $R_{2}(\theta')$ & $R_{3}(\theta')$\\
\midrule
$x$ & $z$ & $y$ & $y$ & $z$ & $y$\\
$y$ & $x$ & $x$ & $z$ & $y$ & $z$\\
$z$ & $y$ & $z$ & $x$ & $x$ & $x$
\end{tabular}
\end{center}

In state $\theta$, each outcome tops some agent's ranking, so $\PO(\theta) = \{x, y, z\}$. In
state $\theta'$, the outcome $x$ is ranked last by every agent and is therefore unanimously
dominated; hence $\PO(\theta') = \{y, z\}$.

We claim that \emph{no mechanism Berge-implements $\PO$.} Suppose for the sake of contradiction that $\Gamma = (M, g)$
does implement it. Then for each pair $(\bar{x}, \bar{\theta})$ with $\bar{x} \in \PO(\bar{\theta})$, there
is a Berge equilibrium $m(\bar{x}, \bar{\theta})$, with $g(m(\bar{x}, \bar{\theta})) = \bar{x}$.
Consider the cross-profile message
\[
  m = \bigl(m_{1}(z, \theta),\; m_{2}(x, \theta),\; m_{3}(z, \theta')\bigr).
\]
This profile is feasible, so $g(m) \in \{x, y, z\}$. We show each possibility leads to a
contradiction.

\begin{itemize}
\item \emph{If $g(m) = x$:} Agent $1$ is holding her message at $m_{1}(z, \theta)$ while
agents $2$ and $3$ jointly change the outcome to $x$. Since $x \, P_{1}(\theta) \, z$, the
opponents of agent $1$ have made her strictly better off than her target $z$, contradicting
that $m(z, \theta)$ is a Berge equilibrium at $\theta$.

\item \emph{If $g(m) = z$:} Agent $2$ is holding her message at $m_{2}(x, \theta)$ while
agents $1$ and $3$ change the outcome to $z$. Since $z \, P_{2}(\theta) \, x$, the opponents
of agent $2$ have made her strictly better off than her target $x$, contradicting that
$m(x, \theta)$ is a Berge equilibrium at $\theta$.

\item \emph{If $g(m) = y$:} Agent $3$ is holding her message at $m_{3}(z, \theta')$ while
agents $1$ and $2$ change the outcome to $y$. Since $y \, P_{3}(\theta') \, z$, the opponents
of agent $3$ have made her strictly better off than her target $z$, contradicting that
$m(z, \theta')$ is a Berge equilibrium at $\theta'$.
\end{itemize}

Every case yields a contradiction, so no such mechanism exists. \qed
\end{example}

\medskip

The obstruction is worth pausing on. Under strategic altruism, the equilibrium at $(z, \theta)$
protects agent $1$: her opponents cannot push the outcome above $z$ for her. Similarly, the
equilibrium at $(x, \theta)$ protects agent $2$ from any outcome above $x$, and the
equilibrium at $(z, \theta')$ protects agent $3$ from anything above $z$. When the agents mix
their equilibrium messages (which is a feasible message profile), the mechanism must produce a single outcome that violates none of
these three protections simultaneously (weakly worse than $z$ for agent $1$, weakly worse
than $x$ for agent $2$, and weakly worse than $z$ for agent $3$). But $\bigcap_{i}
L_{i}(x^{i}, \theta^{i}) = L_{1}(z, \theta) \cap L_{2}(x, \theta) \cap L_{3}(z, \theta') =
\{z\} \cap \{x, y\} \cap \{x, z\} = \emptyset$: no such outcome exists, and the mechanism breaks.

It is instructive to contrast this with Nash equilibrium. Under Nash, each agent is protecting
only herself: no agent's opponent is constrained to act in that agent's interest. The
cross-profile message does not impose the same simultaneous restrictions, and indeed $\PO$ is
Nash-implementable whenever $n \geq 3$ \citep{maskin1999nash}. Now, while the fact that a Berge equilibrium need not be Pareto efficient in some games is well-known,\footnote{See for example \cite{haller2024berge}.} our result is far more general. We find that there are cases where \textit{no game form can support the whole set of Pareto efficient outcomes as Berge equilibria}.

The example above raises a natural follow-up question. Since implementing the \emph{entire} weak Pareto
correspondence is too demanding, perhaps a planner could implement a \emph{selection} from
it. Indeed, in most cases, weak Pareto efficiency alone is too poor as a normative criterion and other considerations, such as no-envy, are imposed on top of it. The following results show that this avenue is almost entirely
closed.

\subsection{The King-Maker Rule}
\label{subsec:kingmaker}

Example~\ref{ex:wp} shows that strategic altruism cannot guarantee even weak Pareto
efficiency. One might still hope that, among efficient rules, some that \emph{share} power
rather than concentrate it could survive. A natural test case is the \emph{king-maker
rule} of \citet{maskin1999nash}.

Fix a distinguished agent, say agent $1$, the king-maker, and let $F^{KM}$ select, at
each state, every outcome that is top-ranked by some \emph{other} agent:
\[
  F^{KM}(\theta) = \bigl\{ x \in X : x \text{ is the most-preferred outcome of some }
  j \in N \setminus \{1\} \bigr\}.
\]
The mechanism \citet{maskin1999nash} proposes for it is disarmingly simple: agent $1$
names one of the others as ``king,'' and the king then chooses the outcome freely. The rule
is appealing on three counts. It is weakly Pareto efficient, since an outcome that tops some
agent's ranking can never be unanimously dominated. It is of course non-dictatorial: no
single agent's favorite is always chosen, and with three or more agents the rule typically
selects several outcomes (the favorites of several different people). And, crucially, 
\citet{maskin1999nash} shows it \emph{is} implementable in Nash equilibrium whenever there
are three or more agents. Self-interested players, left to the king-maker game, produce
exactly these outcomes. Strategically altruistic players, on the other hand, do not.

\begin{example}[Failure of the king-maker rule]\label{ex:km}
Let $N = \{1, 2, 3\}$ with agent $1$ as king-maker, let $X = \{a, b, c\}$, and take the
single state $\theta$ with preferences (most preferred first):
\[
  \begin{array}{ccc}
    R_{1}(\theta) & R_{2}(\theta) & R_{3}(\theta)\\[3pt]
    \hline\\[-8pt]
    a & a & b\\
    c & b & c\\
    b & c & a
  \end{array}
\]
Here $F^{KM}(\theta) = \{a, b\}$, since agent $2$ top-ranks $a$ and agent $3$ top-ranks
$b$. Suppose some mechanism $\Gamma$ Berge-implemented $F^{KM}$, and consider the
cross-profile $\sigma = \bigl((b,\theta),(b,\theta),(a,\theta)\bigr) \in \Sigma^{F^{KM}}$,
in which agents $1$ and $2$ each target the selected outcome $b$ while agent $3$ targets
the selected outcome $a$. By the cross-profile bound (Lemma \ref{lem:cross}), the
cross-profile outcome $w$ would have to satisfy
\[
  w \in L_{1}(b,\theta) \cap L_{2}(b,\theta) \cap L_{3}(a,\theta)
    = \{b\} \cap \{b,c\} \cap \{a\} = \emptyset,
\]
which is impossible. Hence no mechanism Berge-implements $F^{KM}$. \qed
\end{example}

The contradiction is straightforward to read off the preference table. Agent $1$ ranks $b$
last, so the only outcome she weakly prefers to her target $b$ is $b$ itself, hence the
equilibrium protection pins the cross-profile outcome to $b$. Agent $3$ ranks $a$ last,
so her protection pins it to $a$. No outcome is simultaneously weakly below both $b$ (for
agent $1$) and $a$ (for agent $3$), and no mechanism can satisfy both
protections at once.

Self-interested agents are spared from the above constraints because they protect no one. In the king-maker
game, the designated king simply chooses, and a deviation by any other player changes
nothing. When agents are strategically altruistic, each player's least-preferred
\emph{selected} outcome becomes a constraint, and here two of those constraints are
irreconcilable.

This example already carries the paper's thesis in miniature: \emph{Berge implementation
is strictly more demanding than Nash implementation.} There exists an efficient and
non-dictatorial rule that a society of self-interested agents can implement and a society of
strategically altruistic ones cannot. The dictatorship theorem, to which we now turn,
shows that the king-maker rule is not an isolated case. 

\subsection{The Dictatorship Theorem}
\label{subsec:dictatorship}

We are now ready to state the paper's main result:

\begin{theorem}[Berge dictatorship]\label{thm:dictator}
If a SCR $F$ satisfies the Pareto property and is Berge-implementable on $\Thst$, then $F$
is dictatorial.
\end{theorem}

\subsection{Proof of Theorem \ref{thm:dictator}}
\label{subsec:proof}

Suppose $F$ satisfies the Pareto property and is Berge-implemented by $\Gamma = (M, g)$.
Assume, for contradiction, that $F$ is \emph{not} dictatorial.

\medskip
\noindent\textbf{Step 0: Every agent is overruled somewhere.}

Since no agent is a dictator, for each $i \in N$ there is a state--outcome pair $(x^{i},
\theta^{i})$ that ``overrules'' $i$, that is, $\SU_{i}(x^{i}, \theta^{i}) \neq \emptyset$ and
$x^{i} \in F(\theta^{i})$. Collecting these pairs gives us a profile $\sigma =
(x^{i}, \theta^{i})_{i \in N} \in \Sigma^{F}$ in which every agent's target is strictly
improvable, that is, $\SU_{i}(x^{i}, \theta^{i}) \neq \emptyset$ for every $i$.

\medskip
\noindent\textbf{Step 1: No common improvement.}

The whole proof lies on the following strengthening of the cross-profile bound.

\begin{claim}\label{cl:nocommon}
Let $\sigma = (x^{i}, \theta^{i})_{i \in N} \in \Sigma^{F}$ with
$\SU_{i}(x^{i}, \theta^{i}) \neq \emptyset$ for every $i$. Then
\[
  \bigcap_{i \in N} \SU_{i}(x^{i}, \theta^{i}) = \emptyset.
\]
\end{claim}

In words, the agents' targets cannot all be improved upon by the same outcome. If they could,
there would be a single alternative strictly preferred by every agent to their own current
target; a kind of unanimous improvement that strategic altruism, combined with efficiency,
rules out.

\begin{proof}[Proof of Claim~\ref{cl:nocommon}]
Recall that $w = w(\sigma) = g(m^{\sigma})$ is the cross-profile outcome.

\smallskip\noindent\emph{(i) $w$ lies weakly below every target.} By
Lemma \ref{lem:cross} (the cross-profile bound),
$w \in \bigcap_{i} L_{i}(x^{i}, \theta^{i})$.

\smallskip\noindent\emph{(ii) Constructing an auxiliary state.} Since the domain is unrestricted,
choose $\theta \in \Thst$ so that for every $i \in N$,
\[
  \underbrace{\SU_{i}(x^{i}, \theta^{i})}_{\text{ranked on top}}
  P_{i}(\theta)\;
  w
  \;P_{i}(\theta)\;
  \underbrace{L_{i}(x^{i}, \theta^{i}) \setminus \{w\}}_{\text{ranked at the bottom}}.
\]
The three blocks partition $X$ (since $\SU_{i}(x^{i}, \theta^{i}) = X \setminus
L_{i}(x^{i}, \theta^{i})$ and $w \in L_{i}(x^{i}, \theta^{i})$ by (i)), so we may order them
however we like within each block. By construction, $L_{i}(x^{i}, \theta^{i}) \subseteq
L_{i}(w, \theta)$ for every $i$.

\smallskip\noindent\emph{(iii) The cross-profile message is a Berge equilibrium at $\theta$.}
Fix any agent $i$ and hold her message at $m_{i}(x^{i}, \theta^{i})$. The outcomes her
opponents can induce are $O^B_{i}(x^{i}, \theta^{i}) \subseteq L_{i}(x^{i}, \theta^{i})
\subseteq L_{i}(w, \theta)$, so none of them is strictly better for $i$ than $w$ at $\theta$.
Since $i$ was arbitrary, $m^{\sigma}$ is a Berge equilibrium at $\theta$ with outcome $w$.
By implementation, $w \in F(\theta)$.

\smallskip\noindent\emph{(iv) Contradiction via efficiency.} Suppose
$\bigcap_{i} \SU_{i}(x^{i}, \theta^{i}) \neq \emptyset$ and take any outcome $\bar{x}$ in
this intersection. By (ii), $\bar{x}$ sits in every agent's top block at $\theta$, so
$\bar{x} \, P_{i}(\theta) \, w$ for all $i$. The outcome $w$ is therefore unanimously
dominated and $w \notin \PO(\theta)$. But $w \in F(\theta) \subseteq \PO(\theta)$ by the
Pareto property, a contradiction. Hence, the intersection must be empty.
\end{proof}

\medskip
\noindent\textbf{Step 2: Blocking outcomes.}

Since $\SU_{i}(x^{i}, \theta^{i}) \neq \emptyset$ for every $i$, we may choose, for each $i
\in N$, an outcome
\[
  \gamma(i) \in \SU_{i}(x^{i}, \theta^{i})
\]
that $i$ strictly prefers to her own target. By Claim~\ref{cl:nocommon} these choices are
\emph{not} all equal: if $\gamma(i) = x^{\ast}$ for every $i$, then $x^{\ast}$ would belong
to $\bigcap_{i} \SU_{i}(x^{i}, \theta^{i})$, which is empty. Partition $N$ into cells
according to the value of $\gamma$, writing
\[
  N_{c} = \{ i \in N : \gamma(i) = c \}
\]
for the cell of agents whose blocking outcome is $c$, and let $K \subseteq N$ contain
exactly one representative per cell. Relabelling the representatives as
$1, \ldots, |K|$, the outcomes $\gamma(1), \ldots, \gamma(|K|)$ are the distinct
\emph{blocking outcomes}; write $G = \{\gamma(j) : j \in K\}$ with $|K| \geq 2$. All indices
on representatives below are understood modulo $|K|$.

\medskip
\noindent\textbf{Step 3: A Condorcet state and the final contradiction.}

Choose $\theta \in \Thst$ in which every agent ranks the blocking outcomes above everything
else:
\[
  G P_{i}(\theta)\; X \setminus G \quad \text{for all } i \in N.
\]
No outcome outside $G$ is Pareto optimal at $\theta$, so the Pareto property gives
\begin{equation}
  F(\theta) \subseteq G. \tag{$\ast$}\label{eq:inG}
\end{equation}

We now arrange preferences \emph{within} $G$ to derive a contradiction.

\smallskip\noindent\emph{The case $|K| = 3$ (illustrative).} Take three blocking outcomes
$\gamma(1) = a$, $\gamma(2) = b$, $\gamma(3) = c$, one agent per cell ($N = \{1, 2, 3\}$).
Arrange the rankings of $\{a, b, c\}$ at $\theta$ as a Condorcet cycle:
\begin{center}
\begin{tabular}{ccc}
$P_{1}(\theta)$ & $P_{2}(\theta)$ & $P_{3}(\theta)$\\
\midrule
$c$ & $a$ & $b$\\
$b$ & $c$ & $a$\\
$a$ & $b$ & $c$
\end{tabular}
\end{center}
Each agent's worst outcome in $G$ is her own blocking outcome ($a$ for agent $1$, $b$ for
agent $2$, $c$ for agent $3$), while her blocking outcome is the top choice of the next agent
cyclically.

By~\eqref{eq:inG} and nonemptiness, $F(\theta)$ contains some blocking outcome; say $a =
\gamma(1) \in F(\theta)$. Now consider $b = \gamma(2)$. By the cyclic arrangement:
\begin{itemize}
\item[(a)] every agent outside the cell $N_{\gamma(2)} = \{2\}$ strictly prefers $b$ to $a$
(agent $1$: $b P_1(\theta) a$; agent $3$: $b P_3(\theta) a$); and
\item[(b)] agent $2$'s target $x^{2}$ already lies strictly below $b$, since
$b = \gamma(2) \in \SU_{2}(x^{2}, \theta^{2})$ by construction.
\end{itemize}
Form the profile $\sigma' = \bigl((a, \theta), (x^{2}, \theta^{2}), (a, \theta)\bigr) \in
\Sigma^{F}$ (legitimate because $a \in F(\theta)$ and $x^{2} \in F(\theta^{2})$). Then $b$
lies in the strict upper-contour set of every agent in $\sigma'$:
\[
  b \in \SU_{1}(a, \theta), \quad b \in \SU_{2}(x^{2}, \theta^{2}), \quad b \in \SU_{3}(a, \theta).
\]
This means $b \in \bigcap_{j} \SU_{j}(\cdot)$ with every set nonempty, contradicting
Claim~\ref{cl:nocommon}. The cases where $b$ or $c$ is selected are symmetric.

\smallskip\noindent\emph{General $|K| \geq 2$.} Within $G$, assign each representative $i
\in K$ the cyclic ranking
\[
  \gamma(i-1) \; P_{i}(\theta)\; \gamma(i-2) \; P_{i}(\theta)\; \cdots \; P_{i}(\theta)\;
  \gamma(i+1) \; P_{i}(\theta)\; \gamma(i) \pmod{|K|}
\]
so that $\gamma(i)$ is representative $i$'s worst outcome in $G$; assign every agent in the
same cell as $i$ the same ranking. By~\eqref{eq:inG} and nonemptiness some $\gamma(i) \in
F(\theta)$. Consider $\gamma(i+1)$ (indices modulo $|K|$). By the cyclic arrangement:
every agent outside cell $N_{\gamma(i+1)}$ strictly prefers $\gamma(i+1)$ to $\gamma(i)$; for
every agent $j$ inside that cell, $\gamma(i+1) = \gamma(j) \in \SU_{j}(x^{j}, \theta^{j})$
by construction. Form the profile $\sigma' \in \Sigma^{F}$ that gives each agent outside
$N_{\gamma(i+1)}$ the pair $(\gamma(i), \theta)$ and each agent $j$ inside it the pair
$(x^{j}, \theta^{j})$. Then $\gamma(i+1)$ belongs to every agent's strict upper-contour set
in $\sigma'$, with every such set nonempty---contradicting Claim~\ref{cl:nocommon}.

\medskip
This contradiction exhausts all cases. Hence $F$ must be dictatorial, which completes the proof. \qed

\subsection{Discussion}
\label{subsec:discussion}

Theorem \ref{thm:dictator} and Examples~\ref{ex:wp}--\ref{ex:km} together paint a stark
picture. Under strategic altruism, the set of implementable and efficient social goals is
precisely the set of dictatorial rules. A few remarks are in order.

\begin{remark}[Comparison with Nash implementation]\label{rem:nash}
The result is surprising in light of the classical Nash impossibility in \citet{hurwicz1978construction}, which establishes that when $n = 2$,
any weakly Pareto-efficient and Nash-implementable SCR on the full domain of strict preferences must
be dictatorial. For three or more agents, however, this impossibility \emph{disappears} with Nash while, in our case, the impossibility persists for \emph{any} number of agents. Thus, what drives our impossibility is the behavioral assumption itself, not the size of the group.
\end{remark}

\begin{remark}[Domain]\label{rem:domain}
The proof uses only that $\Thst \subseteq \Theta$. Theorem \ref{thm:dictator} therefore holds on any domain that contains the unrestricted strict domain, such as a domain of weak orders that contains $\Thst$.
\end{remark}

\begin{remark}
The social dilemma literature celebrates strategic altruism because
it resolves the conflict between individual and collective rationality \emph{within} a given game \citep{COLMAN2011166, haller2024berge}. Our results show that it creates a new and equally severe
problem at the level of institutional design: it forecloses efficient,
non-dictatorial rules that self-interested societies can operate. Our pessimistic conclusion runs sharply against the intuitive appeal of altruism
as a foundation for institution design.
\end{remark}

\begin{remark}[Positive results under domain restrictions]\label{rem:domain_positive}
The impossibility results hold on the unrestricted domain, which places no constraints on
how agents rank the available outcomes. However, we can recover some optimism in restricted domains that arise naturally in economic
environments. Consider a pure one-to-one matching environment
in which staying single is the worst possible outcome for every agent. This condition
holds, for instance, in marriage markets where remaining unmatched is universally
unacceptable. On this domain, the set of stable matchings is Berge-implementable.\footnote{In this domain, \cite{tadenuma1998implementable} study the class of Nash-implementable matching rules.} The reason is transparent in terms of the cross-profile
bound: when staying single is everyone's least preferred option, every agent's lower-contour
set at any matched outcome is large enough, as it includes the single outcome. The cross-profile
outcome can therefore always be found inside the intersection of these sets. 
\end{remark}

\section{Berge Implementation Is Harder Than Nash}
\label{sec:berge-implies-nash}

The impossibility results of Section~\ref{sec:impossibility} establish that strategic
altruism is a restrictive behavioral protocol for institutional design. Many efficient and
non-dictatorial rules that a society of self-interested agents can operate lie beyond the reach
of a strategically altruistic society. A natural question that arises at this point is how severe this restriction
really is. Could there be social goals that strategically altruistic agents \emph{can}
implement but self-interested ones cannot? 
 
The answer is no. Berge implementability is strictly more demanding than Nash
implementability: every social goal achievable under strategic altruism is also achievable
under self-interest, but not vice versa. This is summarized in our proposition below:
 
\begin{proposition}[Berge implies Nash]\label{prop:BimpliesN}
If a SCR $F$ is Berge-implementable, then $F$ is Nash-implementable. Moreover, when
$n = 2$, the converse also holds: $F$ is Berge-implementable if and only if it is
Nash-implementable.
\end{proposition}

\subsection{Why Berge Is Harder: The Identification Problem}
\label{subsec:identification}
 
To see intuitively why Berge implementation is the more demanding standard, it is useful to
understand where the extra difficulty lies. As explained in Section \ref{sec:model}, full implementation has two requirements: (i)
\emph{existence}, which dictates that every socially optimal outcome must be supported by some equilibrium;
and (ii) \emph{uniqueness}, which dictates that every equilibrium must produce a socially optimal outcome.
 
In Nash implementation with three or more agents, existence comes essentially for free. The
following simple mechanism works: each agent $i$ is asked to report an outcome-state pair $(x^i,\theta^i)$ such that $x^i$ is socially optimal at state $\theta^i$ and, if at least $n-1$ agents make the same announcement, the announced outcome is implemented. Note that unanimous announcement of $(x,\theta)$ at state $\theta$ is a Nash equilibrium since no single agent can profitably deviate (the outcome function is insensitive to a single deviation). This logic is behind the classical results in Nash implementation: if a
single agent deviates from unanimous agreement, the planner can identify her as the
``deviator'' by tailoring the outcome function accordingly.

In Berge implementation, this logic however breaks down. The problem is that Berge
equilibrium requires each agent to fix her \emph{own} message while her opponents optimize
for her. When checking whether a cross-profile message is a Berge equilibrium, each player is potentially playing a message from a \emph{different} scenario, and
there is no way to identify who is the ``deviator''. To see this, suppose that $m_1=m_2=(x,\theta)$, while $m_3=(y,\theta')$. While in the case of Nash players the planner can potentially identify 3 as the deviator from the unanimous announcement $(x,\theta)$, with Berge players, she cannot: it could be that agents 1 and 2 deviated
from $(y, \theta')$, or it could be that agents 2 and 3 deviated from $(x, \theta)$ and so on. Thus, with Berge players, the planner cannot tailor the outcome function to punish
the right agents, because she cannot tell which ones they are. This identification failure is
precisely what makes the cross-profile bound of Lemma \ref{lem:cross} binding, and it is why
Berge implementability requires a strictly stronger condition on $F$ than Nash
implementability does.
 
\subsection{Outline of the Proof}
\label{subsec:proof-BimpliesN}
 
The argument proceeds via the characterizing conditions for each implementation notion; the
formal statements and complete proofs can be found in Appendix \ref{app:proof-prop}. Here
we describe the logic in plain terms.

For our characterization, we will be using two important notions. One is the set $O^B_i(x^i,\theta^i)$, as defined in Section \ref{sec:model}, which represents the opportunity set of $i$'s opponents while $i$ plays her part $m_i(x^i,\theta^i)$ of a Berge equilibrium according to $\sigma$. The second is an analogue construction from the Nash viewpoint: $O^N_i(x,\theta)$ is $i$'s opportunity set when her opponents are playing their part $m_{-i}(x,\theta)$ of a Nash equilibrium $m(x,\theta)$ that supports $x$ in state $\theta$.\footnote{Note that, in the case of Nash, one only needs to define $O_i^N$ for a given Nash equilibrium $(x,\theta)$ and not for an arbitrary mix of equilibrium messages $m^{\sigma}$ as we have done in the case of Berge. This is due to the importance of cross-profile messages in the latter case. Moreover, note that these two constructions are in general different.} Formally, given a Nash equilibrium $m(x,\theta)$ such that $x = g(m(x,\theta))$,
\begin{center}

$O^N_i(x,\theta) = \bigl\{g(m_i', m_{-i}(x,\theta)) : m_i' \in M_i\bigr\}$
  
\end{center}

\paragraph{What characterizes Nash implementability.}
\citet{moore1990nash} show that, with three or more agents, $F$ is Nash-implementable if and
only if it satisfies their Condition $\mu$, which is stated in terms of the existence of a family of
sets $\CN_i$. Informally: there is a range $B$ of attainable outcomes, and for each $x \in
F(\theta)$ a family of sets $\CN_i(x,\theta) \subseteq B$ with $x$ top-ranked by every agent
within her own set at state $\theta$, such that for every alternative state $\theta'$ the rule respects three
closure requirements. First, if $x$ is each agent's most-preferred outcome within her own set
at state $\theta'$, then $x \in F(\theta')$. Second, if some outcome $c$ is agent $i$'s
most-preferred within her set and simultaneously every \emph{other} agent's most-preferred
within the whole range $B$ at $\theta'$, then $c \in F(\theta')$. Third, if an outcome $d$ is
top-ranked by \emph{all} agents within the whole range $B$ at $\theta'$, then it too must be
selected, that is, $d\in F(\theta')$. Now, it is easy to check that, from the necessity point of view, whenever a mechanism Nash-implements $F$, the opportunity sets $O^N_i$ satisfy the requirements for the $\CN_i$ sets.

\paragraph{What characterizes Berge implementability.}
We establish (Theorem \ref{thm:charBerge} in the appendix) that $F$ is Berge-implementable
if and only if it satisfies Condition $\beta^\varepsilon$; only the ``only if'' half of that
equivalence, isolated as Lemma \ref{lem:necessity} in the Appendix, is needed for the
present proposition. This condition includes the same
three-part structure as Condition $\mu$, that is, it relies on the existence of a family of sets $\CB_i$ that satisfy the above three requirements. However, it carries a \emph{fourth}
part with no counterpart in Condition $\mu$ (for the case of three or more agents). For \emph{any} number of agents $n$, and for any
profile of ``announcements'' $\sigma=(x^i,\theta^i)_{i \in N}$ with $x^i \in F(\theta^i)$, the family of Berge sets must have a non-empty common intersection,
\[
  \bigcap_{i \in N} \CB_i(x^i, \theta^i) \neq \emptyset,
\]
together with an associated (monotonicity-type) closure requirement on any outcome in that intersection. Specifically, if an outcome on this intersection weakly improves on everyone's preference when switching to a new state, then it should be selected as socially optimal at that state. This is
exactly the cross-profile bound of Lemma \ref{lem:cross} paired with a monotonicity-type requirement, re-expressed as a condition on
$F$ itself. In Condition $\mu$ this requirement appears only in the
two-agent characterization; under Berge behavior it binds for every group size. Finally, note that, in the Berge case as well, given a mechanism that Berge-implements a SCR $F$, the opportunity sets $O^B_i$ serve as the $\CB_i$ sets.

\paragraph{The logical relationship and the proof.}
The link between the two conditions runs through a single fact, established in the Appendix:
\emph{whenever $F$ is Berge-implementable, the implementing mechanism's Berge opportunity sets
$O^B_i$ can serve as the Nash opportunity sets $O^N_i$.} Intuitively, the structure that Berge equilibrium imposes is strong enough to supply exactly the closure
properties that Condition $\mu$ demands of a Nash mechanism and, additionally, carries an even stricter requirement, that of the cross-profile bound. In other words, since Condition $\beta^\varepsilon$ satisfies all three parts of Condition $\mu$ (and adds the fourth, cross-profile,
requirement on top), it is a strict strengthening of Condition $\mu$. The implication
therefore flows directly:
\[
  \begin{aligned}
    F \text{ is Berge-implementable}
      &\;\xRightarrow{\;\text{Lemma~\ref{lem:necessity}}\;}\;
      F \text{ satisfies  Condition }\beta^\varepsilon\\
      &\;\Rightarrow\;
      F \text{ satisfies Condition }\mu\\
      &\;\xRightarrow{\;\text{\cite{moore1990nash}}\;}\;
      F \text{ is Nash-implementable}.
  \end{aligned}
\]
 
\emph{Two-agent equivalence.} It is noteworthy that when $n = 2$, the gap closes from the other side. With only
two agents, agent $i$'s Berge opportunity set $O^B_i$ coincides, after relabeling the two players, with
$j$'s Nash opportunity set $O^N_j$ in a suitably permuted mechanism. The cross-profile
intersection requirement of Condition $\beta^\varepsilon$ then matches the additional
two-agent condition that \citet{moore1990nash} require for Nash implementation
(part~($\mu$iv) of their Condition $\mu_2$ or, equivalently, Condition $\beta$ of
\citealt{dutta1991necessary}). The two characterizations become equivalent,
so Nash implementability implies Berge implementability. The full formal
argument, including the explicit permutation of message spaces that constructs the Berge
mechanism from any Nash mechanism for the case of two agents, is in Appendix~\ref{app:proof-prop}.
 
\subsection{The Altruism Tax}
\label{subsec:altruism-tax}
 
Proposition \ref{prop:BimpliesN} has a stark economic interpretation. Think of the
set of Berge-implementable SCRs as the collection of social goals that a strategically
altruistic society can achieve through institutional design, and the set of
Nash-implementable SCRs as the collection achievable by a self-interested society. The
proposition says the former is a strict subset of the latter and illustrates the idea of the \emph{altruism tax}: paradoxically, moving from egoistic to altruistic decision-making restricts what a social planner can achieve.



This bears on an old verdict. \citet{shubik1961review}, with his negative review, dismissed Berge's
monograph as offering nothing to economists. However, the subsequent
literature in Economics has restored the relevance of Berge equilibrium with applications in Cournot and Bertrand models, as well as in social dilemmas
\citep{COLMAN2011166, haller2024berge} and, on that terrain, its importance is
decisive. Our results though identify a different terrain,
on which Shubik's verdict survives in a far sharper form: as a behavioral premise for institutional design, strategic
altruism is not merely uninteresting, but strictly dominated.

Moreover, as shown in Sections~\ref{subsec:pareto-impossible}
and~\ref{subsec:kingmaker}, the gap between the two sets is populated by economically important rules. The weak Pareto
rule and the king-maker rule are Nash-implementable for $n \geq 3$ \citep{maskin1999nash} yet,
as we have shown, not Berge-implementable. The dictatorship theorem of
Section~\ref{subsec:dictatorship} establishes that the only Berge-implementable efficient
rules are dictatorial, while Nash implementation is consistent with a rich variety of
efficient, non-dictatorial rules for any group of three or more agents. The altruism tax is
therefore not a technical curiosity confined to exotic examples; it has a bite in cases of interest to economists, political theorists and social scientists in general.
 
The finding is surprising since, a priori, there was no need to expect such nestedness in the implementation notions, especially given that there is no such nestedness in the equilibrium concepts.\footnote{Note that nestedness of equilibrium notions does not in general imply nestedness of the corresponding implementation notions. Both subgame-perfect and strong Nash equilibria are subsets of the Nash equilibria, yet Nash implementability implies subgame-perfect implementability (a simultaneous-move game form has no proper subgames, and thus its subgame-perfect and Nash equilibria coincide \citep{moore1988subgame}), whereas it does not imply strong Nash implementability, since the set of strong Nash equilibria may be empty or a strict subset of the Nash set \citep{dutta1991implementation}.} The result also stands in sharp contrast to popular narratives that equate altruism with social
progress. Indeed, our results show the exact opposite: from a mechanism-design perspective, the strategic
altruism encapsulated in Berge behavior has no institutional value whatsoever. A social
planner can always do at least as well, and typically strictly better, when agents
pursue their own interests, because self-interest leaves the planner with more room to
steer agents' incentives towards the desirable goal.

\section{Conclusions}
\label{sec:conclusion}
 
Our results establish three interconnected findings about the
institutional consequences of strategic altruism, formalized as
Berge equilibrium. First, no mechanism can implement the weak
Pareto rule in
Berge equilibrium.
Second, any social choice rule that is both weakly Pareto efficient
and Berge-implementable on the full domain of strict orderings, must be dictatorial that is, it unconditionally imposes
one agent's preferred outcome on all of society. Third, Berge implementability is
strictly more demanding than Nash implementability: every social goal that can be achieved by institutional design in a society of strategic
altruists can also be
achieved in a society of self-interested agents, but the converse
is false. The source of all three results is the same structural
feature of Berge equilibrium, the cross-profile bound, which forces
any implementing mechanism to simultaneously honor every agent's
equilibrium protection across all possible states of the world.
Efficiency then concentrates these obligations into the preferences
of a single agent and dictatorship emerges. 

Our results suggest that \citet{shubik1961review}'s verdict was right all along for a reason he did not give: it is not that strategic altruism lacks economic content, but that \textit{its content is confined to the analysis of given games, and does not extend to the design of the games themselves}. We conclude the paper with a promising follow-up research agenda:

\paragraph{In search of a tax-free behavioral norm.}
Our results establish that strategic altruism carries an
institutional cost that self-interest does not. This raises an important question: does there exist a
behavioral norm that simultaneously resolves social dilemmas
\emph{and} avoids the altruism tax?
 
The answer requires a possible move away from \emph{other-regarding}
cooperation (each agent devoting her strategic effort to
advancing others' payoffs) towards \emph{self-regarding}
cooperation in the sense of \citet{UNVEREN20231}. The
structural root of the altruism tax is the Berge counterfactual:
each agent holds her own message fixed while her opponents vary,
which is essentially the Nash counterfactual applied to others'
payoffs. It is this shared unilateral structure that generates
the cross-profile bound and, with it, the dictatorship result.
\emph{Kantian equilibrium} \citep{roemer2010kantian}, by
contrast, replaces the unilateral counterfactual with a
universal one: each agent asks what would happen if everyone
varied their strategy proportionally. No agent holds her own
message fixed; instead, deviations are collective thought
experiments. This different counterfactual is what
allows Kantian equilibrium to resolve the Prisoner's Dilemma and
related social dilemmas. Moreover, as
\citet{roemer2010kantian} has already shown, Kantian
equilibrium can implement efficient, non-dictatorial allocation
rules in production economies, a positive result that stands in
sharp contrast to ours. Characterizing the
full set of Kantian-implementable social choice rules on the
unrestricted domain with not necessarily symmetric game forms, and determining whether Kantian
implementation escapes the altruism tax altogether, is the most
direct open problem that our results bring into focus.
 
A similar direction concerns \emph{team reasoning}
\citep{sugden2011mutual, sugden2015team}, in which agents do not
optimize individually at all but instead ask ``what should we
collectively do?'' and implement their component of the
collectively optimal action. Team reasoning also resolves social
dilemmas, and its relationship to mechanism design remains
unexplored. With our characterization theorem, we identify which cross-profile
property distinguishes Berge implementation from Nash implementation; developing the
analogous characterization for team-reasoning-based
implementation would clarify whether the altruism tax is specific
to the Berge counterfactual or is a feature of other
behavioral norms.

\bibliographystyle{apalike}
\bibliography{berge}

\appendix

\section{Characterization of Berge Implementation and Proof of Proposition~\ref{prop:BimpliesN}}
\label{app:proof-prop}
 
This appendix states Condition $\mu$ of \citet{moore1990nash} and Condition $\beta^\varepsilon$
formally; establishes that Condition $\beta^\varepsilon$ is necessary for Berge implementation
(Lemma~\ref{lem:necessity}) and, together with sufficiency, that it characterizes it
(Theorem~\ref{thm:charBerge}); and completes the proof of
Proposition~\ref{prop:BimpliesN}, which invokes only the necessity Lemma. To present the conditions more compactly, we need a piece of additional notation. For any set
$S \subseteq X$, agent $i$, and state $\theta$, let
\[
  MA_i(S, \theta) \equiv \{ x \in S : x \, R_i(\theta) \, y \text{ for all } y \in S \}
\]
denote the set of maximal outcomes in $S$, i.e.
her \emph{most-preferred} outcomes within $S$.

\subsection*{Condition $\mu$ \citep{moore1990nash}}
 
For completeness, we review the characterization of Nash implementation against which our
result is stated. \citet{moore1990nash} show that, for $n \geq 2$, $F$ is Nash-implementable
if and only if it satisfies the following condition:
 
\begin{adefinition}[Condition $\mu$]\label{def:mu}
A SCR $F$ satisfies \emph{Condition $\mu$} if there exist a set $B\supseteq \bigcup_{\theta}F(\theta)$ and, for each
$i \in N$, $\theta \in \Theta$, and $x \in F(\theta)$, a set $\CN_i(x,\theta) \subseteq B$
with $x \in MA_i(\CN_i(x,\theta), \theta)$, such that for every $\theta' \in \Theta$:
\begin{enumerate}
  \item[\upshape($\mu$i)] if $x \in \bigcap_{i \in N} MA_i(\CN_i(x,\theta), \theta')$, then
    $x \in F(\theta')$;
  \item[\upshape($\mu$ii)] if $c \in MA_i(\CN_i(x,\theta), \theta') \cap \bigl[\bigcap_{j \neq i}
    MA_j(B, \theta')\bigr]$ for some $i \in N$, then $c \in F(\theta')$;
  \item[\upshape($\mu$iii)] if $d \in \bigcap_{i \in N} MA_i(B, \theta')$, then
    $d \in F(\theta')$.
\end{enumerate}
\end{adefinition}
 
Part ($\mu$i) is a strengthening of Maskin monotonicity, part ($\mu$ii) is a weakening of no-veto power and ($\mu$iii) a weakening of unanimity. For two agents, the relevant condition is Condition $\mu_2$, which augments
($\mu$i)--($\mu$iii) with a fourth, intersection-type requirement:

\begin{enumerate}
  \item[\upshape($\mu$iv)] for any $\sigma=\bigl((x^1,\theta^1),(x^2,\theta^2)\bigr)
    \in \Sigma^F$:
  \begin{itemize}
    \item there exists $w(\sigma) \in \bigcap_{i \in N} \CN_i(x^i,\theta^i)$; and
    \item for all $\bar{\theta} \in \Theta$, if $w(\sigma) \in \bigcap_{i \in N}
      MA_i\bigl(\CN_i(x^i,\theta^i), \bar{\theta}\bigr)$, then
      $w(\sigma) \in F(\bar{\theta})$.
  \end{itemize}
\end{enumerate}
 
\subsection*{Condition $\beta^\varepsilon$}

Below we state the condition that characterizes Berge implementation:

\begin{adefinition}[Condition $\beta^\varepsilon$]\label{def:betaeps}
A SCR $F: \Theta \rightrightarrows X$ satisfies \emph{Condition $\beta^\varepsilon$} if
there exist a set $Y \supseteq \bigcup_\theta F(\theta)$; for every $i \in N$ and
every $(x, \theta)$ with $x \in F(\theta)$, a set $\CB_i(x,\theta) \subseteq Y$
with $x \in \CB_i(x,\theta) \subseteq L_i(x,\theta)$; and a selection
$w \colon \Sigma^F \to X$; such that the following hold.
 
\begin{enumerate}
  \item[\upshape(i)] \emph{Cross-profile selection.} For every
    $\sigma = (x^i, \theta^i)_{i \in N} \in \Sigma^F$,
    \[
      w(\sigma)\in \bigcap_{i \in N} \CB_i(x^i, \theta^i),
    \]
    with $w(\sigma) = x$ whenever $(x^i,\theta^i) = (x,\theta)$ for every $i \in N$;
    and, for every $\theta' \in \Theta$, if
    $\CB_i(x^i, \theta^i) \subseteq L_i\bigl(w(\sigma), \theta'\bigr)$ for all $i \in N$,
    then $w(\sigma) \in F(\theta')$.
 
  \item[\upshape(ii)] \emph{(Stronger version of) Monotonicity.} For any $(x,\theta)$ with
    $x \in F(\theta)$ and any $\theta' \in \Theta$: if $\CB_i(x,\theta) \subseteq L_i(x,\theta')$
    for all $i \in N$, then $x \in F(\theta')$.
 
  \item[\upshape(iii)] \emph{(Weaker version of) no-veto power.} For any $(x,\theta)$ with
    $x \in F(\theta)$ and any $\theta' \in \Theta$: if there exists $i \in N$ and
    $y \in Y$ such that $y \in \CB_i(x,\theta)$, $\CB_i(x,\theta) \subseteq L_i(y,\theta')$,
    and $Y \subseteq L_j(y,\theta')$ for all $j \neq i$, then $y \in F(\theta')$.
 
  \item[\upshape(iv)] \emph{(Weaker version of) unanimity.} For any $\theta' \in \Theta$ and any $x \in Y$:
    if $Y \subseteq L_i(x, \theta')$ for all $i \in N$, then $x \in F(\theta')$.
\end{enumerate}
\end{adefinition}
 
\begin{aremark}
Parts (ii)--(iv) of Condition $\beta^\varepsilon$ are the structural counterparts of
parts~($\mu$i)--($\mu$iii) of Condition $\mu$ (Definition~\ref{def:mu}), obtained by
replacing the Nash sets $\CN_i$ with the Berge sets $\CB_i$:
part~(ii) corresponds to ($\mu$i), part~(iii) to ($\mu$ii), and part~(iv) to ($\mu$iii).
The novel ingredient is part~(i), the cross-profile selection requirement,
which has no counterpart in Condition $\mu$ for $n \geq 3$. For two agents, parts (i)--(iv)
together reduce to Condition $\mu_2$ of \citet{moore1990nash} (equivalently, Condition $\beta$
of \citealt{dutta1991necessary}), whose fourth part is exactly the two-agent intersection
requirement that part~(i) generalizes to an arbitrary number of agents.
\end{aremark}
 
\subsection*{Necessity of Condition $\beta^\varepsilon$}
 
We first isolate the direction of the characterization on which
Proposition~\ref{prop:BimpliesN} rests: any mechanism that Berge-implements $F$ generates,
through its Berge opportunity sets, a family of sets witnessing Condition
$\beta^\varepsilon$.
 
\begin{alemma}[Necessity of Condition $\beta^\varepsilon$]\label{lem:necessity}
If a SCR $F$ is Berge-implementable, then $F$ satisfies Condition $\beta^\varepsilon$.
\end{alemma}
 
\begin{proof}
Suppose $F$ is Berge-implemented by $\Gamma = (M,g)$.
Set $Y = g(M)$ and, for any $i \in N$, $\theta \in \Theta$ and $x \in F(\theta)$, let
$\CB_i(x,\theta) = O^B_i(x,\theta) \subseteq Y$. By definition, $x \in \CB_i(x,\theta)$.
Suppose, for the sake of contradiction, that $\CB_i(x,\theta) \nsubseteq L_i(x,\theta)$.
Then there exists $m'_{-i} \in M_{-i}$ such that $y \equiv g\bigl(m_i(x,\theta),m'_{-i}\bigr)
\notin L_i(x,\theta)$, that is, $y \, P_i(\theta) \, x$. But then $i$'s opponents can make
her strictly better off than $x$, so $m(x,\theta) \notin BE(\Gamma,\theta)$---a
contradiction. Hence $\CB_i(x,\theta) \subseteq L_i(x,\theta)$. Define the selection
$w \colon \Sigma^F \to X$ by $w(\sigma) \equiv g(m^{\sigma})$.
 
\emph{$\beta^{\varepsilon}$(i).} Fix $\sigma \in \Sigma^F$ and form the cross-profile message
$m^\sigma = (m_i(x^i,\theta^i))_{i \in N}$. By definition,
$g(m^\sigma) \in \CB_i(x^i,\theta^i)$ for every $i$, so
$w(\sigma) \in \bigcap_i \CB_i(x^i,\theta^i)$. When all pairs coincide at $(x,\theta)$,
$m^\sigma = m(x,\theta)$ is a Berge equilibrium that supports outcome $x$, so $w = x$. Now suppose
$\CB_i(x^i,\theta^i) \subseteq L_i(w(\sigma),\theta')$ for all $i$. Then for any agent $i$, holding
her message at $m_i(x^i,\theta^i)$, her opponents can only attain outcomes in
$\CB_i(x^i,\theta^i) \subseteq L_i(w(\sigma),\theta')$, so $m^\sigma$ is a Berge equilibrium at
$\theta'$ with $g(m^\sigma) = w(\sigma) \in F(\theta')$.
 
\emph{$\beta^{\varepsilon}$(ii).} Fix $(x,\theta)$ with $x \in F(\theta)$ and suppose
$\CB_i(x,\theta) \subseteq L_i(x,\theta')$ for all $i$. Then $m(x,\theta)$ remains a Berge
equilibrium at $\theta'$, so $x \in F(\theta')$.
 
\emph{$\beta^{\varepsilon}$(iii).} Fix $(x,\theta)$ with $x \in F(\theta)$ and suppose $y \in
\CB_i(x,\theta)$, $\CB_i(x,\theta) \subseteq L_i(y,\theta')$, and $Y \subseteq L_j(y,\theta')$
for all $j \neq i$. Since $y \in \CB_i(x,\theta)$, there exists $m_{-i}'$ such that
$g(m_i(x,\theta), m_{-i}') = y$. Now, for agent $i$, all outcomes reachable by her opponents are
in $\CB_i(x,\theta) \subseteq L_i(y,\theta')$. Moreover, for any $j \neq i$, all outcomes in $Y$,
including all reachable outcomes by her opponents, are in $L_j(y,\theta')$. Hence $(m_i(x,\theta), m_{-i}')$
is a Berge equilibrium at $\theta'$, and $y \in F(\theta')$.
 
\emph{$\beta^{\varepsilon}$(iv).} Fix $x \in Y$ and suppose $Y \subseteq L_i(x,\theta')$ for all $i$.
Since $x \in Y = g(M)$, there exists $m$ with $g(m) = x$. For every agent $i$ and every
$m_{-i}'$, $g(m_i, m_{-i}') \in Y \subseteq L_i(x,\theta')$, so $m$ is a Berge equilibrium
at $\theta'$, and $x \in F(\theta')$. Hence $F$ satisfies Condition $\beta^\varepsilon$.
\end{proof}
 
\subsection*{Characterization theorem}
 
Condition $\beta^\varepsilon$ is not only necessary but also sufficient, so that it
characterizes Berge implementation exactly as Condition $\mu$ characterizes Nash
implementation.
 
\begin{atheorem}[Characterization of Berge implementation]\label{thm:charBerge}
A SCR $F$ is Berge-implementable if and only if it satisfies Condition $\beta^\varepsilon$.
\end{atheorem}
 
\begin{proof}
\noindent\textbf{(Only if.)} This is Lemma~\ref{lem:necessity}.
 
\medskip\noindent\textbf{(If.)} Suppose $F$ satisfies
Condition $\beta^\varepsilon$. We construct a mechanism $\Gamma=(M,g)$ that
Berge-implements $F$.

\medskip\noindent\textbf{The mechanism:}
For every $i\in N$ let
\[
  M_i = \bigl\{(x^i,\theta^i,f^i,y^i,n^i)\in Y\times\Theta\times\{NF,F\}
               \times Y\times\mathbb{Z} \;\big|\; x^i\in F(\theta^i)\bigr\}.
\]
Write $K(m)\equiv\{i\in N\mid f^i=F\}$ for the set of agents who \emph{raise a flag} and
$\sigma(m)\equiv(x^i,\theta^i)_{i\in N}\in\Sigma^F$ for the \emph{reported
cross-profile}. The outcome function $g$ is:
\begin{enumerate}
\item If $K(m)=\emptyset$: \;(a) if all agents report the same pair $(x,\theta)$,
  then $g(m)=x$; \;(b) otherwise, $g(m)=w(\sigma(m))$.
\item If $\emptyset\neq K(m)\neq N$: \;(a) if $N\setminus K(m)=\{i\}$ and all
  $j\in K(m)$ report the same $(x,\theta,y)$, then $g(m)=y$ if
  $y\in\CB_i(x^i,\theta^i)$ and $g(m)=x^i$ otherwise; \;(b) otherwise,
  $g(m)=w(\sigma(m))$.
\item If $K(m)=N$: $j^{*}=\min_{i\in N}\{\operatorname{argmax}_{j}n^j\}$ and $g(m)=y^{j^{*}}$.
\end{enumerate}
Existence of $w(\sigma(m))$ in $\bigcap_i\CB_i(x^i,\theta^i)$ in cases 1(b) and 2(b)
is guaranteed by Condition $\beta^\varepsilon$(i). Every rule places $g(m)$ in $Y$
since all components $x^i,y^i\in Y$ and $\CB_i\subseteq Y$.

\textbf{Existence:} Let us first show that $F(\theta)\subseteq g(BE(\Gamma,\theta))$.
Fix any $x \in F(\theta)$ and let $m_i=(x,\theta,NF,\cdot,\cdot)$ for all $i\in N$. Then,
according to rule (1)(a), $g(m)=x$. Suppose that, for some $i\in N$, the agents in
$N\setminus \{i\}$ deviate to $m'_{-i}$. Since $f^i=NF$, the profile $(m_i,m'_{-i})$ can
never fall into rule (3), so it falls into one of the following rules.

\begin{itemize}
    \item Rule (1)(a). Then every agent reports $(x,\theta)$ and
      $g(m_i,m'_{-i})=x\in L_i(x,\theta)$.
    \item Rule (1)(b). Then $g(m_i,m'_{-i})=w(\sigma(m_i,m'_{-i}))\in
      \bigcap_{j\in N}\CB_{j}(x^j,\theta^j)\subseteq \CB_i(x,\theta)\subseteq L_i(x,\theta)$,
      where $(x^j,\theta^j)_{j\in N}=\sigma(m_i,m'_{-i})$ and $(x^i,\theta^i)=(x,\theta)$.
    \item Rule (2)(a). Then $N\setminus K(m_i,m'_{-i})=\{i\}$, so
      $g(m_i,m'_{-i})\in \CB_{i}(x,\theta)\subseteq L_i(x,\theta)$.
    \item Rule (2)(b). Then $g(m_i,m'_{-i})=w(\sigma(m_i,m'_{-i}))\in
      \bigcap_{j\in N}\CB_{j}(x^j,\theta^j)\subseteq \CB_i(x,\theta)\subseteq L_i(x,\theta)$,
      as in rule (1)(b).
\end{itemize}

In every case the deviation to $m'_{-i}$ is not profitable for $i$. Since the choice of
$m'_{-i}$ was arbitrary, $g(m_i,M_{-i})\subseteq L_i(x,\theta)$, and since the choice of
agent $i$ was arbitrary, $g(m_i,M_{-i})\subseteq L_i(x,\theta)$ for all $i\in N$. We conclude
that $m \in BE(\Gamma,\theta)$, and so $x\in g(BE(\Gamma,\theta))$. Since the choice of
$x\in F(\theta)$ was arbitrary, $F(\theta)\subseteq g(BE(\Gamma,\theta))$.

\textbf{Uniqueness:} Next, let us show that $g(BE(\Gamma,\theta))\subseteq F(\theta)$. To this
end, fix any $m\in BE(\Gamma,\theta)$ and suppose that it falls into one of the following
rules.

\begin{itemize}
    \item Rule (1)(a). Then $m_i=(x,\theta',NF,\cdot,\cdot)$ for all $i\in N$, and
    $g(m)=x$. By construction of $M_i$ we have $x\in F(\theta')$, so there is nothing to
    prove if $\theta'=\theta$; suppose therefore that $\theta'\neq\theta$. Fix any $i\in N$
    and any $y\in\CB_i(x,\theta')$, and suppose the agents in $N\setminus\{i\}$ deviate to
    $m'_{-i}$ with $m'_j=(x,\theta',F,y,\cdot)$ for all $j\in N\setminus\{i\}$. Then
    $K(m_i,m'_{-i})=N\setminus\{i\}$, so $(m_i,m'_{-i})$ falls into rule (2)(a), and since
    $y\in\CB_i(x,\theta')$ we get $g(m_i,m'_{-i})=y$. As $y$ was arbitrary,
    $\CB_i(x,\theta')\subseteq g(m_i,M_{-i})$. Since $m\in BE(\Gamma,\theta)$ and $g(m)=x$,
    we have $g(m_i,M_{-i})\subseteq L_i(x,\theta)$, and therefore
    $\CB_i(x,\theta')\subseteq L_i(x,\theta)$. As $i$ was arbitrary, this holds for every
    $i\in N$, and $\beta^{\varepsilon}$(ii) applied to $(x,\theta')$ yields $x\in F(\theta)$.

    \item Rule (1)(b). Then $m_i=(x^i,\theta^i,NF,\cdot,\cdot)$ for all $i\in N$, where
    $(x^i,\theta^i)\neq (x^j,\theta^j)$ for some $i,j\in N$, and $g(m)=w(\sigma(m))$. Fix any
    $i\in N$ and any $y\in \CB_{i}(x^i,\theta^i)$, and suppose the agents in
    $N\setminus \{i\}$ deviate to $m'_{-i}$ with $m'_j=(x^i,\theta^i,F,y,\cdot)$ for all
    $j\in N\setminus \{i\}$. Then $(m_i,m'_{-i})$ falls into rule (2)(a), and
    $g(m_i,m'_{-i})=y$. Since the choice of $y$ was arbitrary,
    $\CB_{i}(x^i,\theta^i) \subseteq g(m_i,M_{-i})$, and since $m\in BE(\Gamma,\theta)$ with
    $g(m)=w(\sigma(m))$, we have $\CB_{i}(x^i,\theta^i) \subseteq
    L_i(w(\sigma(m)),\theta)$. As the choice of agent $i$ was arbitrary, this holds for all
    $i\in N$, and $\beta^{\varepsilon}$(i) implies that $w(\sigma(m))\in F(\theta)$.

    \item Rule (2)(a). Then there exists $i \in N$ with $m_i=(x^i,\theta^i,NF,\cdot,\cdot)$
    such that $m_j=(\hat{x},\hat{\theta},F,\bar{y},\cdot)$ for all $j\neq i$, and
    $g(m)\in \CB_{i}(x^i,\theta^i)$.

    Fix any $y \in \CB_{i}(x^i,\theta^i)$ and suppose the agents in $N\setminus \{i\}$
    deviate to $m'_{-i}$ with $m'_j=(\hat{x},\hat{\theta},F,y,\cdot)$ for all
    $j\in N\setminus \{i\}$. Then $(m_i,m'_{-i})$ falls into rule (2)(a) and
    $g(m_i,m'_{-i})=y$. Since the choice of $y$ was arbitrary,
    $\CB_{i}(x^i,\theta^i) \subseteq g(m_i,M_{-i})$, and since $m\in BE(\Gamma,\theta)$,
    $\CB_{i}(x^i,\theta^i) \subseteq L_i(g(m),\theta)$.

    Now fix any $j\in N\setminus \{i\}$ and any $\bar{z} \in Y$, and suppose agent $i$
    changes $m_i$ into $m'_i=(x^i,\theta^i,F,\bar{z},n^i)$. Then $K(m'_i,m_{-i})=N$, so
    $(m'_i,m_{-i})$ falls into rule (3); by choosing $n^i$ high enough agent $i$ wins the
    integer game and $g(m'_i,m_{-i})=\bar{z}$. Since the choice of $\bar{z} \in Y$ was
    arbitrary, $Y \subseteq g(m_j,M_{-j})$, and since $m\in BE(\Gamma,\theta)$,
    $Y \subseteq L_j(g(m),\theta)$. As the choice of agent $j\neq i$ was arbitrary,
    $Y \subseteq L_j(g(m),\theta)$ for all $j\in N\setminus\{i\}$. Then
    $\beta^{\varepsilon}$(iii) implies that $g(m)\in F(\theta)$.

    \item Rule (2)(b). Then $m_j=(x^j,\theta^j,F,y^j,n^j)$ for all $j\in K(m)$,
    $m_i=(x^i,\theta^i,NF,y^i,n^i)$ for all $i\in N\setminus K(m)$, and $g(m)=w(\sigma(m))$.

    Fix any $i\in N\setminus K(m)$ and any $y \in \CB_{i}(x^i,\theta^i)$, and for each
    $j\in N\setminus \{i\}$ let $m'_j=(x^i,\theta^i,F,y,\cdot)$. Then
    $K(m_i,m'_{-i})=N\setminus\{i\}$, so $(m_i,m'_{-i})$ falls into rule (2)(a) and
    $g(m_i,m'_{-i})=y$. Since the choice of $y$ was arbitrary,
    $\CB_{i}(x^i,\theta^i) \subseteq g(m_i,M_{-i})$, and since $m\in BE(\Gamma,\theta)$,
    $\CB_{i}(x^i,\theta^i) \subseteq L_i(w(\sigma(m)),\theta)$. As the choice of
    $i \in N\setminus K(m)$ was arbitrary, this holds for all $i \in N\setminus K(m)$.

    Now fix any $j \in K(m)$ and any $\bar{z} \in Y$, and suppose that each agent
    $i\in N\setminus K(m)$ changes $m_i$ into $m'_i=(x^i,\theta^i,F,\bar{z},n^i)$. Then every
    agent raises a flag, so the resulting profile falls into rule (3); any such agent $i$
    wins the integer game by choosing $n^i$ high enough, and the outcome is $\bar{z}$. Since
    the choice of $\bar{z} \in Y$ was arbitrary, $Y \subseteq g(m_j,M_{-j})$, and since
    $m\in BE(\Gamma,\theta)$, $Y \subseteq L_j(w(\sigma(m)),\theta)$. As the choice of
    $j\in K(m)$ was arbitrary, $Y \subseteq L_j(w(\sigma(m)),\theta)$ for all $j\in K(m)$;
    since $\CB_{j}(x^j,\theta^j)\subseteq Y$, it follows that
    $\CB_{j}(x^j,\theta^j) \subseteq L_j(w(\sigma(m)),\theta)$ for all $j\in K(m)$.

    Therefore $\CB_{i}(x^i,\theta^i) \subseteq L_i(w(\sigma(m)),\theta)$ for every $i\in N$,
    and $\beta^{\varepsilon}$(i) implies that $w(\sigma(m))\in F(\theta)$.

  \item Rule (3). Then $K(m)=N$, so $m_i=(x^i,\theta^i,F,y^i,n^i)$ for all $i\in N$. Fix any
  $i \in N$ and any $y \in Y$, and suppose each agent $j\neq i$ changes $m_j$ into
  $m'_j=(\cdot, \cdot, F, y,n^j)$. Then $(m_i,m'_{-i})$ falls into rule (3) because
  $K(m_i,m'_{-i})=N$, and some $j\neq i$ wins the integer game by choosing $n^j$ high
  enough, so that $g(m_i,m'_{-i})=y$. Since the choice of $y \in Y$ was arbitrary,
  $Y \subseteq g(m_i,M_{-i})$, and since $m\in BE(\Gamma,\theta)$, $Y \subseteq
  L_i(g(m),\theta)$. As the choice of agent $i\in N$ was arbitrary, $Y \subseteq
  L_i(g(m),\theta)$ for all $i\in N$, and $\beta^{\varepsilon}$(iv) implies that
  $g(m)\in F(\theta)$.

  This exhausts all possible cases and completes the proof.

\end{itemize}

\end{proof}

\subsection*{Proof of Proposition~\ref{prop:BimpliesN}}
 
\begin{proof}[Proof of Proposition~\ref{prop:BimpliesN}]
\noindent\textbf{(Berge $\Rightarrow$ Nash, for any $n\geq 3$).}
Suppose $F$ is Berge-implementable. By Lemma \ref{lem:necessity}, $F$ satisfies
Condition $\beta^\varepsilon$ with the set $Y$ and the Berge sets
$\{\CB_i(x,\theta)\}$. To establish Nash implementability via Condition $\mu$, we must exhibit
a range $B$ and sets $\{\CN_i(x,\theta)\}$ satisfying ($\mu$i)--($\mu$iii). We take
$B = Y$ and set
$\CN_i(x,\theta) \equiv \CB_i(x,\theta)$ for every $i$, $x$ and $\theta$ such that $x\in F(\theta)$. We now verify that the Berge
sets $\CB_i$ satisfy these properties. First, $x \in MA_i(\CB_i(x,\theta),\theta)$ holds
for every $i$, since $\CB_i(x,\theta) \subseteq L_i(x,\theta)$ places $x$ at the top of
$\CB_i(x,\theta)$ under $\theta$. Fix any $\theta'$.
 
\emph{($\mu$i).} Suppose $x \in \bigcap_{i \in N} MA_i(\CB_i(x,\theta), \theta')$, i.e.\
$x$ is each agent's most-preferred outcome within her own Berge opportunity set at $\theta'$.
Then $\CB_i(x,\theta) \subseteq L_i(x,\theta')$ for every $i$, so part~(ii) of
Condition $\beta^\varepsilon$ delivers $x \in F(\theta')$.
 
\emph{($\mu$ii).} Suppose $c \in MA_i(\CB_i(x,\theta), \theta') \cap \bigl[\bigcap_{j \neq
i} MA_j(B, \theta')\bigr]$ for some $i$. Then $c \in \CB_i(x,\theta)$ with $\CB_i(x,\theta)
\subseteq L_i(c,\theta')$ (as $c$ tops $\CB_i(x,\theta)$ at $\theta'$), and $B = Y \subseteq
L_j(c,\theta')$ for every $j \neq i$ (as $c$ tops the whole range for each such $j$). These
are exactly the hypotheses of part~(iii) of Condition $\beta^\varepsilon$, which yields
$c \in F(\theta')$.
 
\emph{($\mu$iii).} Suppose $d \in \bigcap_{i \in N} MA_i(B, \theta')$. Then $Y = B
\subseteq L_i(d,\theta')$ for every $i$, so part~(iv) of Condition $\beta^\varepsilon$ gives
$d \in F(\theta')$.
 
By the
characterization theorem of \citet{moore1990nash}, $F$ is Nash-implementable. (Note that
only the necessity direction of the characterization is invoked, and within it only parts
(ii)--(iv) of Condition $\beta^\varepsilon$; part~(i), the cross-profile
intersection requirement, is not needed for this direction when $n \geq 3$, consistent with
the fact that it is the extra demand Berge places beyond Nash.)

\medskip\noindent\textbf{(Berge $\Leftrightarrow$ Nash, $|N| = 2$).}
We prove both directions directly by means of the permuted mechanism, without relying on
Condition $\mu$ (which characterizes Nash implementation only for $n \geq 3$).
 
\emph{Construction.} Suppose $\Gamma' = (M', g')$ is any mechanism, with
$M' = M_1' \times M_2'$. Define the \emph{permuted mechanism} $\Gamma = (M, g)$ by
$M_1 = M_2'$, $M_2 = M_1'$, and $g(m_1, m_2) = g'(m_2, m_1)$.
 
The key observation is that for any profile $m = (m_1, m_2) \in M$ and any state $\theta$:
\begin{itemize}
  \item The set of outcomes agent~$2$'s deviations can produce while agent~$1$ holds $m_1$
    fixed in $\Gamma$ equals the set of outcomes agent~$1$'s deviations can produce while
    agent~$2$ holds $m_2$ fixed in $\Gamma'$:
    \[
      \{g(m_1, m_2') : m_2' \in M_2\} = \{g'(m_2, m_1') : m_1' \in M_1'\}.
    \]
  \item Symmetrically, the set of outcomes agent~$1$'s deviations can produce while agent~$2$
    holds $m_2$ fixed in $\Gamma$ equals the set of outcomes agent~$2$'s deviations can
    produce while agent~$1$ holds $m_1$ fixed in $\Gamma'$:
    \[
      \{g(m_1', m_2) : m_1' \in M_1\} = \{g'(m_2', m_1) : m_2' \in M_2'\}.
    \]
\end{itemize}
In other words, \emph{the Berge opportunity set of agent~$i$ in $\Gamma$ equals the Nash
opportunity set of agent~$j \neq i$ in $\Gamma'$}, for the same held message. It follows
immediately that:
\begin{equation}\label{eq:permute}
  m = (m_1, m_2) \in BE(\Gamma, \theta)
  \;\iff\;
  \hat{m} = (m_2, m_1) \in NE(\Gamma', \theta).
\end{equation}
Indeed, $m \in BE(\Gamma, \theta)$ requires that the opponents of each agent $i$ cannot
improve upon $g(m)$ for $i$, which by the observation above is exactly the requirement that
$i$ herself cannot improve upon $g'(\hat{m})=g(m)$ in $\Gamma'$, i.e.\ that $\hat{m}$ is a
Nash equilibrium of $\Gamma'$.
 
\emph{Berge $\Rightarrow$ Nash, $|N|=2$.} Suppose $\Gamma$ Berge-implements $F$. Apply the
permutation to obtain $\Gamma'$ with $M_1' = M_2$, $M_2' = M_1$, and $g'(m_1', m_2') =
g(m_2', m_1')$. By~\eqref{eq:permute}, for every $\theta$:
\begin{itemize}
  \item \emph{Existence.} For any $x \in F(\theta)$ there exists $m \in BE(\Gamma, \theta)$
    with $g(m) = x$. Then $\hat{m} = (m_2, m_1) \in NE(\Gamma', \theta)$ and
    $g'(\hat{m}) = g(m) = x$.
  \item \emph{Uniqueness.} If $\hat{m} \in NE(\Gamma', \theta)$, then $m = (\hat{m}_2,
    \hat{m}_1) \in BE(\Gamma, \theta)$, so $g'(\hat{m}) = g(m) \in F(\theta)$.
\end{itemize}
Hence $\Gamma'$ Nash-implements $F$.
 
\emph{Nash $\Rightarrow$ Berge, $|N|=2$.} Suppose $\Gamma'$ Nash-implements $F$. Apply
the permutation to obtain $\Gamma$ with $M_1 = M_2'$, $M_2 = M_1'$, and $g(m_1, m_2) =
g'(m_2, m_1)$. By~\eqref{eq:permute}, for every $\theta$:
\begin{itemize}
  \item \emph{Existence.} For any $x \in F(\theta)$ there exists $\hat{m} \in NE(\Gamma',
    \theta)$ with $g'(\hat{m}) = x$. Then $m = (\hat{m}_2, \hat{m}_1) \in BE(\Gamma,
    \theta)$ and $g(m) = g'(\hat{m}) = x$.
  \item \emph{Uniqueness.} If $m \in BE(\Gamma, \theta)$, then $\hat{m} = (m_2, m_1) \in
    NE(\Gamma', \theta)$, so $g(m) = g'(\hat{m}) \in F(\theta)$.
\end{itemize}
Hence $\Gamma$ Berge-implements $F$ as well and this completes the proof.
\end{proof}

\end{document}